%% file: main.tex
\documentclass{article}
\input{head}

\title{Dynamic Entropy-Encoded Arrays in $O(1)$ Time with Nearly Optimal Space}
\author{
  Guy E. Blelloch\thanks{Carnegie Mellon University. \texttt{guyb@cs.cmu.edu}}
  \and
  Yang Hu\thanks{Tsinghua University. \texttt{y-hu22@mails.tsinghua.edu.cn}}
  \and
  William Kuszmaul\thanks{Carnegie Mellon University. \texttt{kuszmaul@cmu.edu}}
  \and
  Tianxiao Li\thanks{Independent researcher. \texttt{modulo998244353@gmail.com}}
  \and
  Renfei Zhou\thanks{Carnegie Mellon University. \texttt{renfeiz@andrew.cmu.edu}}
}
\date{}

\begin{document}

\maketitle

\begin{abstract}
We show how to implement a dynamic array $A[1, n]$ with symbols from a fixed alphabet $\Sigma$, while supporting $O(1)$-time queries and updates, and using a total space of
\[
  \log \binom{|\Sigma|}{m} + \left(1 + O\left(\frac{\log \log n}{\log n}\right)\right) \cdot \left(\sum_{\sigma \in \Sigma} f_\sigma \log (n / f_\sigma)\right) + n / \polylog n
\]
bits, where $f_\sigma$ denotes the frequency of each symbol $\sigma \in \Sigma$ and $m$ denotes the number of distinct symbols with non-zero frequencies. This resolves a long-standing open question as to whether one can achieve space bounds close to that of arithmetic coding, while supporting $O(1)$-time operations, whenever the entropy is at least $n/\polylog n$.

We also prove a nearly matching space lower bound: up to a factor of $O(\log \log n)$, the entropy-dependent multiplicative overhead of our construction is optimal among $O(1)$-time solutions when $|\Sigma|=O(\sqrt n)$ and the entropy $\sum_{\sigma \in \Sigma} f_\sigma \log (n / f_\sigma)$ lies between $n/\log^{O(1)}n$ and $(1/100)n\log n$. Finally, we present several applications of our results, resolving two open problems having to do with space-efficient dictionaries and filters.
\end{abstract}

\input{intro}

\input{prelim}

\input{prefix-code}

\input{static-arithmetic-code}

\input{arithmetic-code}

\input{arithmetic-code-large-alphabet}

\input{lower-bound}

\input{applications}

\section*{Acknowledgements} 

This research was funded by NSF grant CCF-2504471, NSF grant CCF-2119069, Jane Street Graduate Research Fellowship and MongoDB PhD Fellowship.

\bibliographystyle{alpha}
\bibliography{ref.bib}

\appendix

\input{app_poisson}

\end{document}

%% file: head.tex
\usepackage[T1]{fontenc}
\usepackage[margin=1in]{geometry}
\usepackage{amsfonts,amsmath,amsthm,amssymb,mathtools}  %

\mathtoolsset{centercolon}
\usepackage{xfrac,nicefrac}
\usepackage{mathdots}
\usepackage{mleftright}  %
\let\left\mleft
\let\right\mright

\usepackage{xspace}
\xspaceaddexceptions{]\}}  %
\usepackage{regexpatch}

\usepackage{bm,bbm,dsfont}  %
\usepackage{caption}
\usepackage[normalem]{ulem}
\usepackage{enumitem}

\usepackage{graphicx}
\usepackage{float}
\usepackage{subcaption}  %
\usepackage{tcolorbox}
\usepackage{tikz}
\usetikzlibrary{decorations.pathreplacing}
\usetikzlibrary{calc}
\usetikzlibrary{positioning}
\usetikzlibrary{arrows.meta}

\usepackage[linesnumbered,boxed,ruled,vlined]{algorithm2e}

\SetCommentSty{mycommfont}

\usepackage{thmtools,thm-restate}
\usepackage[colorlinks,citecolor=blue,linkcolor=blue,urlcolor=red]{hyperref}

\theoremstyle{plain}

\newtheorem{theorem}{Theorem}[section]  %
\newtheorem{lemma}[theorem]{Lemma}

\newtheorem{corollary}[theorem]{Corollary}

\newtheorem{claim}[theorem]{Claim}

\theoremstyle{definition}  %



\usepackage[capitalise]{cleveref}
\crefname{algocf}{Algorithm}{Algorithms}
\Crefname{algocf}{Algorithm}{Algorithms}
\crefname{claim}{Claim}{Claims}
\Crefname{claim}{Claim}{Claims}

\usepackage{xurl}

\newfloat{Distribution}{htbp}{loa}
\crefname{Distribution}{Distribution}{Distributions}
\Crefname{Distribution}{Distribution}{Distributions}
\newfloat{Protocol}{htbp}{loa}
\crefname{Protocol}{Protocol}{Protocols}
\Crefname{Protocol}{Protocol}{Protocols}

\SetKwProg{Function}{Function}{:}{}

\SetKwProg{CodeBlock}{}{}{}
\SetKwProg{Repeat}{Repeat}{:}{}

\DeclarePairedDelimiter{\bk}{(}{)}
\DeclarePairedDelimiter{\Bk}{[}{]}
\DeclarePairedDelimiter{\BK}{\{}{\}}

\DeclarePairedDelimiterX\mysetbase[2]{\lbrace}{\rbrace}{#1\,\delimsize\vert\,#2}
\NewDocumentCommand{\myset}{sO{}m m}{%
  \IfBooleanTF{#1}%
    {\mysetbase*{#3}{#4}}%
    {\mysetbase[#2]{#3}{#4}}%
}

\DeclareMathOperator*{\E}{\mathbb{E}}

\let\Pr\PrAux
\DeclareMathOperator{\poly}{poly}

\DeclareMathOperator{\polylog}{polylog}

\renewcommand{\tilde}{\widetilde}

\newcommand{\eps}{\varepsilon}

\renewcommand{\epsilon}{\eps}

\newcommand{\defn}[1]{\emph{\boldmath\textbf{#1}}}

\usepackage{regexpatch}
\makeatletter
\xpatchcmd\thmt@restatable{%
\csname #2\@xa\endcsname\ifx\@nx#1\@nx\else[{#1}]\fi
}{%
\ifthmt@thisistheone
\csname #2\@xa\endcsname\ifx\@nx#1\@nx\else[{#1}]\fi
\else
\csname #2\@xa\endcsname[{Restated}]
\fi}{}{}
\makeatother

\makeatletter
\let\oldparagraph\paragraph
\renewcommand{\paragraph}[1]{%
  \oldparagraph{\boldmath #1}%
}
\makeatother

%% file: intro.tex
\section{Introduction}

Perhaps the most basic problem in data structures is that of storing a dynamic array $A[1, n]$, where each entry $A[i]$ comes from a fixed alphabet $\Sigma=[|\Sigma|]$\footnote{We adopt the convention that $[m]=\{1,\dots,m\}$.}, and where we wish to support efficient updates (change $A[i]$ to some new value $v$) and queries (return $A[i]$). The alphabet itself does not change over time.

Supposing $|\Sigma|$ is a power of two, the dynamic array problem is trivial to solve with $O(1)$-time operations and space $n \log |\Sigma|$ bits.\footnote{Interestingly, when $|\Sigma|$ is not a power of two, it is still possible---albeit with a highly nontrivial construction \cite{dodis2010changing}---to use $\lceil n \log |\Sigma| \rceil$ bits of space (together with $O(\log n)$ precomputed word constants) with $O(1)$-time queries.} Moreover, this space bound is optimal if, say, the entries $A[1], A[2], \ldots, A[n]$ are independent and uniformly random in $\Sigma$. 

A common scenario, however, is that different symbols $\sigma \in \Sigma$ appear with very different frequencies than one another. In such settings, much better space bounds are (at least, in principle) possible. If each character $\sigma \in \Sigma$ appears $f_\sigma$ times---and if one ignores the space needed to store the $f_i$s themselves---then one might hope to use space close to the empirical entropy\footnote{Throughout this paper, we consider only the basic zero-order entropy, with no context used to predict the next symbol.} of the array, given by
\begin{equation}H(A) := \sum_{\sigma \in \Sigma} f_\sigma \log (n / f_\sigma)
\label{eq:zeroentropy}
\end{equation}
bits.

It has been known since the 1970s that, if one does not care about time-efficiency (and, again, if one does not have to store the $f_\sigma$s), one can use arithmetic coding \cite{arithmeticoriginal} to achieve a space bound within two bits of $H(A)$. It has remained open, however, what the best space bounds are for a solution that supports $O(1)$-time operations, and, in particular, whether a bound of the form $(1 + o(1))H(A)$ is possible. This problem, first introduced implicitly by Shannon in 1948 \cite{shannon1948mathematical}, has stood for decades as one of the most basic open questions in space-efficient data structures \cite{shannon1948mathematical,huffman1952method,arithmeticoriginal,mitzenmacher2001compressed,patrascu2008succincter,jansson2012cram,grossi2013dynamic,li2023dynamic}.

\paragraph{This paper: dynamic succinct entropy-encoded arrays. }
The main result of this paper is an affirmative answer to the above question.

\begin{restatable}{theorem}{Main}
    
    \label{thm:arithmetic-code-poly-n}
    Let $c>1$ be a constant. Let $n$ be an integer. Let $w=\Omega(\log n)$ denote the machine-word size. Let $\Sigma=[|\Sigma|]$ be an alphabet, where $|\Sigma|=2^{O(w)}$. There exists a data structure in the virtual memory model that stores an array of symbols $A[1],\dots,A[n]\in \Sigma$, and supports the following operations in constant time in the worst case:
    \begin{itemize}
        \item Update: Given $1\le i\le n$, change $A[i]$ to some $s\in \Sigma$.
        \item Query: Given $1\le i\le n$, return $A[i]$.
    \end{itemize}
    At any point in time, the data structure uses
    \begin{align}
        \label{eq:ourbound}
        \log\binom{|\Sigma|}{m}+\left(1 + O\left(\frac{\log \log n}{\log n}\right)\right) \cdot H(A)+O(n/(\log n)^c)
    \end{align}
    bits of space, where $m$ is the \emph{current} number of distinct symbols in the array. Both the time and space bounds are with high probability in $n$.

\end{restatable}

For a fixed support of size $m$ and fixed positive frequencies $\{f_\sigma\}$, the number of possible arrays is $n!/\prod_\sigma f_\sigma!$. Moreover, the choices arising from different supports are disjoint. Thus, even when the multiset of frequencies is known in advance, a counting lower bound is
\begin{align*}
    \log\binom{|\Sigma|}{m}+\log\bk*{\frac{n!}{\prod_\sigma f_\sigma!}}.
\end{align*}
The second term is at most $H(A)$, which is the entropy benchmark against which we state our result. The exponent in the additive $n/\polylog n$ term can be chosen to be an arbitrarily large constant. Thus, our result achieves the desired $(1+o(1))H(A)$ bound whenever $H(A)=\Omega(n/\log^d n)$ for some constant $d$. Our solution is randomized, and offers all of its guarantees with high probability in $n$.

We further prove---via a reduction from the dynamic dictionary problem \cite{li2023tight}---that the above bound \eqref{eq:ourbound} is nearly optimal in a broad parameter regime.

\begin{restatable}{theorem}{LbMain}
    
    \label{thm:arith_lb}
    Let $n$ be an integer. Let $H_0\in [n/\log^{O(1)}n,(1/100)\cdot n\log n]$ be a parameter of our choice. There exists an alphabet $\Sigma=[|\Sigma|]$, where $|\Sigma|=O(\sqrt{n})$, such that: Suppose that there exists a data structure in the word RAM model with $w=\Theta(\log n)$ that stores an array $a_1,\dots,a_n$ of $n$ symbols from $\Sigma$, and supports updates and queries in worst-case constant time (with high probability over the internal randomness of the dictionary). Then such a data structure must use
    \begin{align*}
        H+\omega(H/\log n)
    \end{align*}
    bits of space, where $H=\sum_{\sigma\in \Sigma}f_\sigma\log(n/f_\sigma)$ is the current entropy of the array, and $f_\sigma$ denotes the current number of occurrences of $\sigma$ in the array. The lower bound holds even when we require that $H\in [0.9H_0,1.1H_0]$ at all times.

\end{restatable}

In the parameter regimes stated in \cref{thm:arith_lb}, the theorem implies that the multiplicative overhead on the $H(A)$ term in our construction is optimal up to an $o(\log \log n)$ factor among all $O(1)$-time solutions. In these regimes, the additive $n/\polylog n$ term is negligible after choosing its exponent sufficiently large. Thus, our upper and lower bounds establish near-optimality in a broad parameter regime.

Along the way, we also consider the (intuitively simpler) problem of encoding a dynamic array of codewords from a known prefix code $C$, where each codeword is at most $O(\log n)$ bits. For this problem, we give a constant-time solution using space $(1 + O(\log \log n / \log n)) \cdot \sum_i |A[i]|$ bits. We also show that this bound is nearly optimal in the regime where $\sum_i |A[i]|=\Omega(n)$.

\paragraph{Applications to dictionaries and filters. }
Finally, we argue that entropy-encoded arrays should be viewed not just as a natural stand-alone problem, but also as a (surprisingly powerful) \emph{tool} for constructing space-efficient solutions to other data-structural problems. 

To demonstrate this point, we discuss in Section \ref{sec:applications} how to use our results to construct new space-efficient dictionaries and filters. In each case, we are able to resolve a long-standing open question by simply taking an existing data structure (that is time-efficient but not yet very space-efficient) and storing it directly in an entropy-encoded array.

We begin by revisiting the problem of constructing a space-efficient hash table mapping keys $k \in U$ to values $v \in V$. This problem, also known as the \emph{dynamic unordered dictionary problem}, is subject to an information-theoretic space requirement of 
\[
S(|U|, |V|, n) = \log \binom{|U|}{n} + n\log |V|.
\]
bits, where $n$ is the current number of keys. A solution is said to be succinct if it uses $(1 + o(1)) \cdot S(|U|, |V|, n)$ bits of space. We present, in Section \ref{sec:applications}, a succinct solution that achieves $O(1)$-time operations, all with high probability in the current value of $n$. This resolves one of the open questions (the ``dynamic resizing'' question) left by Arbitman, Naor, and Segev \cite{arbitman2010backyard} in their 2010 paper on succinct dictionaries with high-probability guarantees. Prior to our construction, the problem had been resolved only with the additional restriction that $n \ll |U| \cdot |V|$ \cite{bender2023iceberg, bender2022optimal}; in the so-called \emph{small-universe} regime, where $n$ and $|U| \cdot |V|$ are close, the only known solutions remained incompatible with resizing \cite{arbitman2010backyard, bender2022optimal, bender2023iceberg}.

We resolve this problem with a remarkably simple construction: When the data structure enters the small-universe regime, we simply transition to an encoding in which we store the data structure as an array of size $|U|$ with values in $V \cup \{\texttt{null}\}$ (the $\texttt{null}$ value means that a key is not present). Although such an array would, in plain text, be highly space-inefficient, if we instead entropy-encode the array, we arrive at an efficient and succinct data structure.

Next, we turn our attention to the problem of storing a succinct fingerprint filter. Given a \defn{false-positive rate} parameter $\epsilon \in (0, 1)$, and setting $\delta = \ln (1 / (1 - \epsilon))$, a \defn{fingerprint filter} \cite{pagh2005optimal,bender2018bloom,bercea2020dynamic,bender2022optimal,kuszmaul2025fingerprint} is a data structure that approximates a set $S$ of up to $n$ keys by first mapping each key $k \in S$ to a random fingerprint in $[n\delta^{-1} + O(1)]$, and then storing the \emph{multiset} of these fingerprints, while supporting efficient insertions, deletions, and membership queries. Such a data structure can be used to perform \defn{approximate-membership queries} on the set $S$, where a query on a key $k$ is answered by checking whether $k$'s fingerprint is in the filter. Such queries are guaranteed to have a false-positive rate of $\epsilon$, meaning that a query on a key $k \in S$ is guaranteed to return true, while a query on a key $k \notin S$ is guaranteed to return false with probability at least $1 - \epsilon$.

In the setting where $\epsilon^{-1} = \omega(1)$, it is known how to construct succinct fingerprint filters with $O(1)$-time operations \cite{pagh2005optimal,bender2018bloom,bercea2020dynamic,bender2022optimal}. When $\epsilon$ is a constant in $(0,1)$, however, it has remained an open question whether one can construct an $O(1)$-time solution that is also succinct \cite{kuszmaul2025fingerprint}.

Using entropy-encoded arrays, we present two fingerprint-filter constructions that are each succinct for constant $\epsilon$. The first construction is to simply store an array of $n \delta^{-1} + O(1)$ counters $C_1, C_2, \ldots$, where $C_i$ counts the occurrences of each fingerprint. We show that, so long as $\epsilon^{-1}=O(\polylog n)$ and $1-\epsilon=\Omega(1)$, we can directly entropy-encode this array in order to arrive at a succinct\footnote{Here ``succinct'' means that the space usage is at most $(1+o(1))$ times the cost of storing the multiset of fingerprints.} fingerprint filter supporting $O(1)$-time operations, with high probability.

Our second fingerprint-filter construction is equally simple: We take a classic solution known as the \emph{quotient filter} \cite{pagh2005optimal,bender2011don}, which \emph{a priori} is not succinct for any value of $\delta^{-1}$, and we store it in an entropy-encoded array. The result is a new filter that is succinct whenever $\epsilon>1/\sqrt{n}$ and $1-\epsilon=\Omega(1)$, while supporting $O(1)$ expected-time operations.

\subsection{Paper Overview}

The technical body of the paper is structured as follows. We prove our main result by presenting a sequence of four increasingly sophisticated data structures, each of which builds on the previous ones to get stronger and stronger guarantees.

\paragraph{\Cref{sec:prefix}: Dynamic Prefix-Coded Arrays.} We begin by considering the following basic problem: Given a \emph{fixed} prefix code $C$, and an array $A[1, n]$ of $O(\log n)$-bit codewords, how space efficiently can we encode the array while still supporting $O(1)$-time queries and updates? The main result of \Cref{sec:prefix} is a succinct solution to this problem, using 
\begin{equation}\left(1 + O\left(\frac{\log \log n}{\log n}\right)\right) \cdot \sum_i |A[i]|
\label{eq:prefixbound}
\end{equation}
bits of space. We also prove later in the paper (Section \ref{sec:lower}) that this bound is nearly tight---namely, that any constant-time solution must use $(1 + \omega(1/\log n)) \cdot \sum_i |A[i]|$ bits of space.

Our basic approach in \Cref{sec:prefix} is as follows. We partition the array into \emph{chunks} of $\polylog n$ entries each. Then, within each chunk, we maintain a (dynamic) partition of entries into \emph{blocks}, where the sum of the sizes of the entries in each block is always roughly $\epsilon \log n$ bits for some small positive constant $\epsilon$.

Because each block stores only $\epsilon \log n$ bits of information, we can use lookup-table techniques to implement operations on a given block in $O(1)$ time. The challenge then is to maintain the blocks in such a way that they can be efficiently accessed and updated. At a high level, our solution to this is to store the blocks (within a given chunk) as the leaves of a B-tree-like structure, where each internal node of the tree stores $\Theta(\log n)$ bits of information. Here, we rely on an insight that has also been used in several recent papers on space-efficient hashing \cite{bender2024modern,bender2022optimal}---namely, that a B-tree with only $\polylog n$ leaves can be implemented using $O(\log \log n)$-bit pointers, allowing for the internal nodes to have fanout $\Theta(\log n / \log \log n)$. This fanout is large enough that the \emph{depth} of the B-tree remains $O(1)$, allowing us to interact with the tree in $O(1)$ time per operation. With this high-level structure, we are able to build a relatively simple data structure achieving \eqref{eq:prefixbound}.

\paragraph{\Cref{sec:fixeddist}: Dynamic Entropy-Encoded Arrays with a Fixed Distribution.} Next we turn to the problem of storing a dynamic entropy-encoded array, but where the space bound we wish to achieve is determined by some \emph{fixed} distribution $D$ over $\Sigma$, rather than by the (dynamically changing) empirical distribution. The main result of \Cref{sec:fixeddist} is a solution to this problem, with $O(1)$-time operations, that uses space
\begin{equation}\left(1 + O\left(\frac{\log \log n}{\log n} \right) \right) \bk*{\sum_{i = 1}^n -\log D(A[i])} + n / \polylog n
\label{eq:fixeddist}
\end{equation}
bits. 

The main difference between the construction in \Cref{sec:fixeddist} and the construction in \Cref{sec:prefix} is the following. In \Cref{sec:prefix}, we broke the elements into blocks, and then encoded \emph{each} element using a fixed prefix code. Now, in \Cref{sec:fixeddist}, we break the elements into blocks of total $D$-entropy $\Theta(\log n)$ and encode each \emph{entire block} using a special prefix code designed to store the block in space close to its total entropy. In this code, a block is partitioned into $O(1)$ intervals, each of $D$-entropy $O(\epsilon\log n)$ unless it consists of a single higher-entropy symbol. By defining the blocks in the right way, and being careful about certain details, we are able to achieve \eqref{eq:fixeddist}.

\paragraph{\Cref{sec:sublinear}: Dynamic Entropy-Encoded Arrays over a Sublinear Alphabet. }Next, we consider the problem of storing a dynamic array $A$ in space close to $H(A)$ (defined above in \eqref{eq:zeroentropy}), but subject to the constraint that the alphabet $\Sigma$ has size at most $n / \polylog n$. (The constraint lets us store up to $\polylog n$ bits of information about each symbol $\sigma \in \Sigma$ without compromising our space efficiency.) Specifically, we are able to construct a dynamic constant-time solution using space
\begin{equation}\left(1 + O\left(\frac{\log \log n}{\log n} \right) \right) H(A) + n / \polylog n + O(|\Sigma| \cdot \log^2 n).
\label{eq:sublingood}
\end{equation}
bits. 

Here, there are two key algorithmic insights that are needed. The first is that, given the data structure already constructed in \Cref{sec:fixeddist} using some distribution $D$, and given a \emph{new} distribution $D'$ that we would like to use in place of $D$, we can rebuild the data structure based on $D'$ (i.e., encode the array using $H_{D'}(A)$ bits instead of $H_D(A)$ bits) \emph{remarkably fast}---in time roughly $O((H_D(A) + H_{D'}(A)) / \log n)$, where $H_D(A)$ and $H_{D'}(A)$ denote the entropy of the array according to distributions $D$ and $D'$, respectively.

This suggests the following high-level structure for an algorithm. Select checkpoints $t_1, t_2, \ldots$ at which the data structure is rebuilt (in a deamortized fashion) according to the empirical distribution $D^{(t_i)}$ at time $t_i$. If at a checkpoint $t_i$ we have $H_{D^{(t_i)}}(A) = R$ for some $R$, then set the next checkpoint $t_{i + 1}$ to occur $\Theta(R / \log n)$ operations in the future. The challenge, then, is to somehow survive between checkpoints, while preserving the space bound \eqref{eq:sublingood}.

To survive between checkpoints, we use a second algorithmic idea. When specifying the distribution used in the $i$-th checkpoint, we actually use a modified distribution $\widetilde{D}^{(t_i)}$ (rather than $D^{(t_i)}$) in which we reserve a $1/\polylog n$ fraction of the density to encode a collection of special \emph{placeholder symbols} with various entropies. These placeholder symbols do not, \emph{a priori}, correspond to any symbols in $\Sigma$. Then, between checkpoints $t_i$ and $t_{i + 1}$, when the frequency of a symbol $\sigma \in \Sigma$ becomes significantly larger than $n\cdot \widetilde{D}^{(t_i)}(\sigma)$, we assign some placeholder symbol $\tau$ to be used in place of $\sigma$; if the frequency of $\sigma$ eventually becomes even much larger still, then we assign another symbol $\tau'$, with $\widetilde{D}^{(t_i)}(\tau') > \widetilde{D}^{(t_i)}(\tau)$, to be used instead, and so on.

If this is done in the right way, then we can get away with encoding the data structure according to distribution $\widetilde{D}^{(t_i)}$ \emph{for the entire time window} $[t_i, t_{i + 1})$. Critically, although the distribution $\widetilde{D}^{(t_i)}$ is unchanging, the mapping of symbols in $\widetilde{D}^{(t_i)}$'s alphabet $\Sigma'$ to the actual alphabet $\Sigma$ evolves continuously, and this is enough for our space bound to survive between rebuilds. 

\paragraph{\Cref{sec:fullresult}: Dynamic Entropy-Encoded Arrays over an Arbitrary Alphabet. }The final challenge is to extend our solution from \Cref{sec:sublinear} to support alphabets of arbitrary size (so long as characters fit in a machine word). This is achieved in \Cref{sec:fullresult}, where we achieve our final space bound of 
\begin{equation}\log \binom{|\Sigma|}{m} + \left(1 + O\left(\frac{\log \log n}{\log n}\right)\right) \cdot H(A) + n / \polylog n
\label{eq:largealph}
\end{equation}
bits.

To achieve this, we begin with a more modest goal, supporting alphabets of size $|\Sigma| = n \polylog n$. Critically, in this parameter regime, we can exploit the following fact: For each symbol $\sigma \in \Sigma$, we are happy to store up to $\polylog n$ copies of the symbol \emph{explicitly} in the array. This is because, if the symbol $\sigma$ were to show up only $\polylog n$ times, its entropy $-\log (f_\sigma / n)$ would be $\log n - O(\log \log n)$ bits, which is within $O(\log \log n)$ of the number of bits needed to encode $\sigma$ in plain text. Making use of this fact, along with a careful strategy for space-efficiently storing metadata (e.g., frequency counts) for each symbol, we are able to achieve \eqref{eq:largealph} for alphabets of size up to $n \polylog n$.

Finally, to support large alphabets, we perform what is essentially a universe reduction step in order to reduce to an alphabet of size $n \polylog n$. For this, we are able to adapt a data structure due to Demaine, auf der Heide, Pagh, and P\v{a}tra\c{s}cu \cite{demaine2006dictionariis}, which maps keys in a set of size $O(n)$ to (stable!) hash codes in $[n \polylog n]$. This allows us to remap keys in order to simulate an alphabet of size $n \polylog n$.

Putting the pieces together in the right way, we are able to achieve the space bound \eqref{eq:largealph} for any alphabet size, which gives the main result of the paper.

\paragraph{\Cref{sec:lower}: Lower Bounds.} Having established our upper bounds, we complement them with nearly matching lower bounds. Here, we set $|\Sigma| = O(\sqrt{n})$, and we prove that any $O(1)$-time dynamic array must use
\[
(1 + \omega(1/\log n)) H(A)
\]
bits of space. This is achieved via a direct reduction to a lower bound previously proven by Li et al.~\cite{li2023tight} for succinct dynamic dictionaries. More precisely, for any prescribed $H_0\in[n/\log^{O(1)}n,(1/100)n\log n]$, the hard instances satisfy $H(A)\in[0.9H_0,1.1H_0]$ throughout the execution.

As noted earlier, we also give a second lower bound for the prefix-code section. For every prescribed $H_0\in[100n,(1/100)n\log n]$, we construct a prefix code and hard instances whose total codeword length remains $\Theta(H_0)$, and show that any $O(1)$-time dynamic array for that code must use space at least
\[
(1 + \omega(1/\log n)) \sum_i |A[i]|
\]
bits. This second lower bound is also proven via a reduction, although the argument (perhaps surprisingly) ends up requiring several additional ideas that the first lower bound does not.

Combined, our lower bounds tell us that, in the parameter regimes just stated, the multiplicative space \emph{redundancies} of the constructions in \Cref{sec:prefix,sec:fullresult} are optimal up to an $o(\log \log n)$ factor.

\paragraph{\Cref{sec:applications}: Applications.} Finally, we complete the paper with \Cref{sec:applications}, which presents the applications to dictionaries and filters already discussed earlier.

\subsection{Prior Work}

The problem of entropy-encoding an array was first introduced implicitly by Shannon in a seminal 1948 paper \cite{shannon1948mathematical}. Shannon showed how to construct a prefix code that can encode the array using space $H(A) + O(n)$ bits. This result was subsequently improved by Huffman \cite{huffman1952method}, who showed how to construct a provably optimal prefix code for a given set of frequencies $f_1, f_2, \ldots, f_{|\Sigma|}$. Although Huffman's construction is optimal for prefix codes, it still may use space as much as $H(A) + O(n)$ bits.

The next major milestone was due to Rissanen, who, in 1976, introduced the notion of \emph{arithmetic coding} \cite{arithmeticoriginal}. Rissanen proved that arithmetic coding uses space at most $H(A) + 2$ bits (assuming the frequencies $f_1, f_2, \ldots$ are already known).

These approaches \cite{shannon1948mathematical, huffman1952method, arithmeticoriginal} support efficient encoding/decoding of the \emph{entire array}, but do not support efficient local queries/updates to a \emph{specific position} in the array. This has proven to be a bottleneck in many real-world applications. A common way around this, albeit with far-from-$O(1)$-time operations, is to break the array into chunks and encode each chunk separately (see, e.g., \cite{pearlman2001trends, xie2006code}).

The next major milestone was a query-efficient construction due to P\v{a}tra\c{s}cu \cite{patrascu2008succincter}, who gave a static array, supporting $O(t)$-time queries, that uses space 
\[
O(|\Sigma| \log n) + H(A) + \frac{n}{((\log n)/{t})^t} + \tilde{O}(n^{3/4}).
\]

Prior to our work, it was only known how to construct efficient dynamic solutions in settings where the alphabet is very small. A line of work by Jansson, Sadakane, and Sung
\cite{jansson2012cram} and Grossi, Raman, Rao, and Venturini \cite{grossi2013dynamic} culminated with a solution supporting $O(1)$-time operations with space $H(A) + n \cdot O((\log |\Sigma| + \log \log n) / \log_{|\Sigma|} n)$. If $\log |\Sigma| = o(\sqrt{\log n})$, then this is within $o(n)$ bits of optimal. In addition to achieving good space bounds for small alphabets, the constructions in \cite{jansson2012cram,grossi2013dynamic} have the further advantage that they are relatively simple, and variations can even be implemented in practice \cite{klitzke2016general}. More recent work by Li, Liang, Yu, and Zhou \cite{li2023dynamic} gives a dynamic solution for constant-size alphabets that uses space $H(A) + n / \polylog n$, but with $\polylog\log n$-time operations.

A closely related problem is that of compressing a length-$n$ \emph{string}, while supporting rank/select/access queries. This has been studied both in the static setting \cite{DBLP:journals/tcs/FerraginaV07,DBLP:conf/soda/GrossiGV03,DBLP:journals/algorithmica/BarbayCGNN14}, and in the dynamic setting \cite{makinen2008dynamic,he2010succinct,jansson2012cram,DBLP:conf/soda/NavarroN13,grossi2013dynamic}, where the goal is to support character insertions/deletions. Notably, Navarro and Nekrich give a solution that supports $O(\log n / \log \log n)$-time operations while using space at most
\[
O(|\Sigma| \log^{1 + \epsilon} n) + \sum_{i \in \Sigma} f_i \log (n / f_i) + O(n),
\]
so long as $|\Sigma| \le \poly(n)$.

Finally, a more basic question than entropy-encoding an array is the problem of storing an array of objects, where each object has a different size. If there are $n$ objects whose sizes sum to $m$ (and that each fit in a machine word), then it is known how to construct solutions that support $O(1)$-time queries and updates, while using space $O(n) + (1 + \epsilon) m$ bits, for some arbitrarily small positive constant $\epsilon \in (0, 1)$ \cite{blandford2008compact, poyias2017compact}; as well as solutions that use $O(n \log \log n) + m$ bits \cite{jansson2012cram}, provided each object is at most $O(\log n)$ bits. An immediate consequence of these results is that, for any \emph{fixed} prefix code $C$, it is possible to maintain a dynamic array with entries encoded by code $C$ using space $O(n) + (1 + \epsilon) m$, where $m$ is the encoding length of the array. A natural question is whether one can hope to reduce this space usage to $(1 + o(1)) m + o(n)$, while still supporting $O(1)$-time operations. This question is answered in the affirmative in Section \ref{sec:prefix} of the current paper, which, in turn, serves as a warmup for several of the other more advanced constructions in the paper.

%% file: prelim.tex
\section{Preliminaries}\label{sec:prelims}

All of our data structures operate in the word-RAM model with some machine-word size $w$. In most applications, $w = \Theta(\log n+\log|\Sigma|)$, but to maximize generality, we will typically present results parameterized by an arbitrary $w = \Omega(\log n+\log|\Sigma|)$. We will model memory as an infinite tape, and measure the space usage of a data structure by the size of the prefix of that tape that the data structure uses.

We will typically use $\Sigma$ to denote the alphabet of characters, and $A$ to denote the array of size $n$ that we are storing. For notational convenience, we will often refer to the $i$-th character $A[i]$ in the array simply by $a_i$. For a given character $\sigma$, we will tend to use $f_\sigma$ to denote the (current) number of occurrences of that character in the array. 

Given a distribution $D$ over $\Sigma$, we will refer to the quantity $-\log D(\sigma)$ as the \defn{$D$-entropy} of a character $\sigma \in \Sigma$, and to the quantity $\sum_i -\log D(a_i)$ as the \defn{$D$-entropy} of the array $A$. As a shorthand, we will often use $H$ to denote the $D$-entropy of the entire array, where $D$ is the \defn{empirical distribution}, given by $D(\sigma) = f_\sigma / n$. Note that all logarithms are base 2 in this paper. We will use $\ln$ for natural logarithms.

Finally, when discussing the memory management of an algorithm, it will sometimes be helpful to refer to the notion of a virtual memory (VM) \cite{li2023dynamic}. A VM simulates a tape of memory that some component of the algorithm uses. The following lemma \cite{li2023dynamic} formalizes the idea that, with standard algorithmic techniques, one can ``simulate'' having access to many separate VMs, even though in reality they must be stored together on a single memory tape. This simulation allows one to access/edit each VM in $O(1)$ time, while also supporting an allocation/release operation on each VM which increases/decreases the length of the VM by up to one machine word.
\begin{lemma}[{\cite[Remark 5.2]{li2023dynamic}}]
    \label{lem:polylog-concat}
    Let there be $B$ VMs of $\ell_1,\dots,\ell_B$ bits respectively ($\ell_i$ changes over time), where $\ell_i\le L$ always holds. When the word size $w=\Omega(\log (BL))$, we can simulate all these (``logical'') VMs using one (``physical'') VM of $\bk*{\sum_{i=1}^{B}\ell_i}+O(B\sqrt{Lw}+Bw)$ bits, such that each word access/edit or allocation/release in any logical VM takes constant time to complete.
\end{lemma}

%% file: prefix-code.tex
\section{Dynamic Prefix-Coded Arrays}\label{sec:prefix}

In this section, we describe our algorithm for the most basic setting---storing a dynamic array of codewords from a fixed prefix code.

\begin{theorem}
    \label{thm:prefix-code-alg}
    Let $n$ be an integer and let $\epsilon>0$ be a constant parameter of our choice. Let $w=\Omega(\log n)$ denote the machine-word size, such that $n\ge w^{10}$. Let $\Sigma\subseteq\{0,1\}^{*}$ be a prefix code (that is, for any distinct $s_1,s_2\in \Sigma$, $s_1$ is not a prefix of $s_2$), such that any codeword $s\in \Sigma$ has length $|s|\in [1,L_{\text{max}}]$, where $L_{\text{max}}=\Theta(w)$ is a parameter. There exists a data structure in the virtual memory model that stores an array of codewords $a_1,\dots,a_n\in \Sigma$, and supports the following operations in constant time in the worst case:
    \begin{itemize}
        \item $\text{Update}(i,s)$: Given $1\le i\le n$, change $a_i$ to $s\in \Sigma$.
        \item $\text{Query}(i)$: Given $1\le i\le n$, return $a_i$.
    \end{itemize}
    The data structure uses $H\cdot (1+O((\log w)/w))$ bits of space, where $H$ is defined as $\sum_{i=1}^{n}|a_i|$. The data structure assumes access to a lookup table that depends only on the prefix code $\Sigma$ and that uses $2^{\epsilon w}$ bits of space.
\end{theorem}

In this section, as an abuse of notation, when we refer to the \defn{entropy} of a set of entries in the array, we mean the total lengths of all codewords in the set. (We will revert to using the word ``entropy'' in more standard ways in subsequent sections.)

\subsection{A Two-Level Structure}

Our data structure uses a two-level structure to store the entries. We partition the array into disjoint \defn{chunks}, each containing exactly $w^{10}$ entries. Within each chunk, we further group the entries into \defn{blocks}, such that the total entropy of a block (that is, the total lengths of the codewords) is $\Theta(w)$.

Within each chunk, we will use a resizable B-tree to maintain the block structure. At a high level, the leaves of this B-tree will be the blocks themselves, and the internal nodes will each store $\Theta(w / \log w)$ pointers that are each encoded with $O(\log w)$ bits. Across all the B-trees (one per chunk) these pointers will collectively contribute $H\cdot O((\log w)/w)$ bits of redundancy.

Formally, we partition the array into chunks of $w^{10}$ entries each, and construct a separate VM for handling each chunk. These VMs are concatenated using \cref{lem:polylog-concat}, which uses $O(\sqrt{w^{11}\cdot w}+w)=O(w^6)$ bits per VM (the VM for each chunk has size at most $O(w^{10}\cdot w)$). Note that the last chunk may have fewer than $w^{10}$ entries. In this case, we merge it with the second-to-last chunk, so that every chunk contains $\Theta(w^{10})$ entries. It now remains to construct a data structure for each chunk.

Within each chunk, we further partition the entries into disjoint blocks. Each block satisfies that the total length of all codewords in it $\sum_i |a_i|$ is in the range $[L_{\max},4L_{\max}]=\Theta(w)$ (recall that $L_{\max}$ is the maximum length of any codeword). 

Throughout the execution, when the length of an entry $a_i$ changes, we split/merge blocks to ensure that the total length of every block is in the right range. It is not hard to see that for each update, we only need to split/merge a constant number of blocks.

\subsection{Managing Memory with a Tiny B-tree}
\label{sec:B-tree}

As noted earlier, to store each block, we borrow a high-level idea that has emerged in recent years in the hash-table literature \cite{bender2024modern, bender2022optimal}---that we can use external-memory data structures (and specifically B-trees) as part of an in-memory data structure, where each machine word in memory acts as a ``block'' for the B-tree.

Our high-level approach is to use a dynamically resizable B-tree to maintain the block structure within each chunk. The B-tree stores a set of key-value pairs, where each key corresponds to the starting point of a block, and the value associated to a key is an $O(w)$-bit string, which is the concatenation of all codewords in the block. When given the index $i$ of an entry, the B-tree returns the identity of the block containing $i$, as well as the concatenation of codewords.

In the following, we show that it is possible to implement such a B-tree using $O(\log w)$ bits of extra space per key (in addition to storing the values). 

\begin{lemma}
    \label{lem:B-tree}
    Let $w$ be the machine word size, and let $m=\poly w$. Let $\epsilon>0$ be a constant of our choice. There exists a data structure for storing a set $S$ of keys from $[m]$, where each key $i$ is associated with a value $s_i$ which is a bit string of length $O(w)$, supporting the following operations in constant time:
    \begin{itemize}
        \item $\text{Insert}(i,s)$: Insert $i$ to $S$, and let its value be $s_i=s$. It is guaranteed that $i$ did not exist in $S$.
        \item $\text{Delete}(i)$: Delete key $i$ from $S$. It is guaranteed that $i$ existed in $S$.
        \item $\text{Query}(i)$: Find the predecessor of $i$ in $S$ (denote by $i'$), and return $(i',s_{i'})$ as the answer.
    \end{itemize}
    At any point in time, the data structure uses
    \begin{align*}
        \sum_{i\in S}\bk*{|s_i|+O(\log w)}+O(m^{1/2}w^2)
    \end{align*}
    bits of space. It also assumes access to a lookup table of size $2^{\epsilon w}$, which only depends on $m$.
\end{lemma}

Note that, since we use \cref{lem:B-tree} on chunks containing $m=\Theta(w^{10})$ entries, and the number of blocks in a chunk is always at least $m/L_{\text{max}}=\Omega(w^{9})$, the extra $O(m^{1/2}w^{2})$ term is negligible compared with the number of blocks.

\begin{proof}
    Let $\delta$ be a sufficiently small positive constant (as a function of $\epsilon$ and $\log_w m$), let $B=\lceil \delta w/\log w\rceil$, and let $h=\log_B m+O(1)=O(1)$. We construct a B-tree of $h$ levels, where each node represents the key-value pairs for the keys in some sub-interval of $[m]$. 
    
    \paragraph{Logical structure.} The nodes are defined in a bottom-up way:
    \begin{itemize}
        \item The level-$h$ (i.e., leaf) nodes directly correspond to the input keys in $S$. The \defn{content} of a level-$h$ node is a key-value pair $(i, s_i)$. 
        \item For $1\le j<h$, each level-$(j+1)$ node is the child of a level-$j$ node, and the number of children of each level-$j$ node is in the range $[B,2B]$. In the exceptional case where the number of level-$(j+1)$ nodes is smaller than $B$, we would have exactly one level-$j$ node.

        The \defn{content} of a level-$j$ node is an $O(w)$-bit string, containing the following information about each child: the sub-interval of $[m]$ corresponding to the child (this is $O(\log m) = O(\log w)$ bits), and a pointer pointing to the physical address of the child. We will describe later how to encode each pointer with $O(\log w)$ bits, so that the total amount of content in the node is at most $2B \cdot O(\log w) = O(w)$ bits. Moreover, if the constant $\delta$ (used in the definition of $B$) is selected to be small enough, then we can conclude a slightly stronger bound, namely that the size of each internal node is actually at most $\epsilon w /2$ bits.
    \end{itemize}
    Queries and updates are implemented in a standard way: For each operation, we first traverse the B-tree to find the content of the leaf node. Then, if we are inserting or deleting a key, we may merge/split a constant number of nodes to keep the number of children of each node in the range $[B,2B]$.

    \paragraph{Efficiently interacting with internal nodes.} In order for the $B$-tree to be efficient, we must be able to efficiently interact with internal nodes. Specifically, given a bit string of length $w$ which represents some internal node $x$, and given a key $k \in [m]$, we must be able to determine in time $O(1)$ which child subtree (if any) contains $k$. 

    Here, we make use of the fact that each internal node is at most $\epsilon w / 2$ bits. It follows that, with a lookup table of size $2^{\epsilon w / 2} \poly(w) \ll 2^{\epsilon w}$, we can directly implement a query to an internal node with a table lookup. This is what allows us to traverse internal nodes in constant time. 

    In addition to querying internal nodes, one must also be able to update them efficiently. This can be done either with another lookup table of size $2^{\epsilon w/2} \poly(w)$, or with standard bit-manipulation techniques.

    \paragraph{Physical structure: Internal nodes.} We use one VM to store the content of the internal nodes. This VM is an array of $N$-bit strings, for some $N=O(w)$, where each string is used to store the content of one node. Note that each node is given a fixed number of bits to store its content, regardless of the number of children. The pointer to a node is then simply an index in the VM (which is $O(\log w)$ bits), indicating the entry where the node content is stored.

    Updates are handled as follows: When allocating space for a new node, we increment the size of the VM by $N$ bits, and store the node content in the new entry. When deleting a node $x$, we move the last entry of the VM (say storing node $y$) to the entry where $x$ was stored, and decrement the size of the array by $N$ bits. 
    
    One difficulty is that, after moving $y$, we need to update the pointer in the content of $y$'s parent, so that it correctly points to the physical address of $y$. To implement this, we first read the content of $y$ to know the interval that it corresponds to, as well as the level of $y$ (say $j$). Then, we search on the B-tree to find the level-$(j-1)$ node that contains $y$ (which has to be $y$'s parent), and update the pointer in its content. Note that we cannot directly store a pointer to $y$'s parent in the content of $y$, since if we did, then after a node is moved, we would need to update the contents of all its children, which would be too time-consuming.

    \paragraph{Physical structure: Leaf nodes.} Storage of the leaf nodes is almost the same as the internal nodes. However, we need to be careful with our space usage, since we can only waste $O(\log w)$ bits per leaf node. That is, we can no longer assign the same amount of space to each node.
    
    To store the leaf contents, we construct not one but $O(w)$ VMs, where the $k$-th VM is responsible for storing the leaf contents $s$ for which the length of $s$ is \emph{exactly} $k$. Each VM is maintained similarly to the VM of the internal nodes. The pointer to the address of a leaf node (stored in the contents of level-$(h-1)$ nodes) consists of the length of its content $s$, as well as the actual address of $s$ in the $|s|$-th VM.
    
    Similar to the internal nodes, when the location for a leaf content changes, we also search on the B-tree to find its parent, and update its pointer.

    Finally, these $O(w)$ VMs (as well as the VM for the internal nodes) are concatenated using \cref{lem:polylog-concat}.

    \paragraph{Total space usage.} The space usage of the B-tree consists of the following parts:
    \begin{itemize}
        \item For each key $i\in S$ (with associated value $s_i$), the content of its corresponding leaf node contains $|s_i|+O(\log w)$ bits, and is stored optimally.
        \item Each internal node uses $O(w)$ bits of space. The total number of internal nodes is
        \begin{align*}
            \bigg\lceil\frac {|S|}B\bigg\rceil+\bigg\lceil\frac {|S|}{B^2}\bigg\rceil+\cdots+\bigg\lceil\frac {|S|}{B^{h-1}}\bigg\rceil=O({|S|}/B),
        \end{align*}
        where $B=\Theta(w/\log w)$. Therefore, they use $O({|S|}\cdot \log w)$ bits in total, which is $O(\log w)$ bits per key.
        \item In the end, we apply \cref{lem:polylog-concat} to concatenate $O(w)$ VMs, where the size of each VM is at most $O(m\cdot w)$. According to \cref{lem:polylog-concat}, we waste $O(m^{1/2}w^2)$ bits due to concatenation. \qedhere
    \end{itemize}
\end{proof}

\subsection{Decoding Block Contents Using a Lookup Table}

Finally, it remains to decode a codeword $a_i$ from the content of the block containing $i$, which is the concatenation of all codewords in the block. Let $\epsilon>0$ be a constant, where we assume for simplicity that $(1/2)\epsilon\cdot w$ is an integer.

The straightforward idea is to construct a lookup table of size $2^{O(\epsilon\cdot w)}$, which, when given a string of length $O(\epsilon\cdot w)$, decomposes it into a sequence of codewords (except the last one which may be a prefix of a codeword). This lookup table can then be used to decompose the content in chunks of $O(\epsilon\cdot w)$ bits.

This approach works when the codewords are all short. However, if some codeword $a_i$ is longer than $\epsilon\cdot w$ bits, then our lookup table does not work. Indeed, in order to handle these long codewords, the size of the lookup table would have to be as large as $2^w$, which is unacceptable to us. 

To fix this, we slightly modify the prefix code $\Sigma$ before running the algorithm: Let $\Sigma'$ denote the modified code, which is obtained from $\Sigma$ by replacing each codeword $s$ with $|s|>(1/2)\epsilon\cdot w$ by $s_{1\dots (1/2)\epsilon\cdot w}\circ \text{str}(|s|)\circ s_{(1/2)\epsilon\cdot w+1\dots |s|}$, where the string $\text{str}(|s|)$ is the binary representation of $|s|$ using $\lceil\log L_{\text{max}}\rceil$ bits (possibly with leading zeros). That is, the length of every codeword longer than $(1/2)\epsilon\cdot w$ increases by $O(\log w)$. When actually running the algorithm, each codeword is replaced by its counterpart in $\Sigma'$.

The benefit of using $\Sigma'$ is that we no longer need lookup tables for decoding long codewords: If we have read a prefix of $(1/2)\epsilon\cdot w+O(\log w)$ bits from the content and the current codeword has not ended, then we directly learn the length of the current codeword. Decoding of a block's content proceeds $(1/2)\epsilon\cdot w$ bits at a time, where each time the lookup table either successfully decomposes the current string, or it claims that the current string is a prefix of some long codeword. In the latter case, we simply skip to the end of the current codeword. Such a lookup table only uses $2^{(1/2)\epsilon\cdot w}\ll 2^{\epsilon w}$ bits. This concludes our construction.

%% file: static-arithmetic-code.tex
\section{Dynamic Entropy-Encoded Arrays with a Fixed Distribution}
\label{sec:fixeddist}

In this section, we consider the problem of storing an array of symbols according to some distribution $D$, where $D$ is fixed over time, and we allow the data structure to use $-\log D(\sigma)$ bits when storing the symbol $\sigma$.

This setting can be seen as a special case of the arithmetic-coding setting, where the empirical distribution $D(\sigma)=f_\sigma/n$ is fixed. In the next section, we will solve the general case, where the empirical distribution changes over time, using the algorithm in this section as a subroutine.

\begin{theorem}
    \label{thm:static-arithmetic-code-alg}
    Let $c\ge 4,\epsilon>0$ be constants. Let $n$ be an integer. Let $w=\Omega(\log n)$ denote the machine-word size. Let $\Sigma=[|\Sigma|]$ (i.e., $\Sigma$ contains the first $|\Sigma|$ positive integers) denote the alphabet, where $|\Sigma|=2^{O(w)}$. Let $D$ be a distribution over $\Sigma$, where $D(\sigma)>0$ for any $\sigma\in \Sigma$. There exists a data structure that stores an array of symbols $a_1,\dots,a_n\in \Sigma$, and supports the following operations in worst-case constant time:
    \begin{itemize}
        \item Update: Given $1\le i\le n$, change $a_i$ to some $s\in \Sigma$.
        \item Query: Given $1\le i\le n$, return the symbol $a_i$.
    \end{itemize}
    The data structure uses
    \begin{align*}
        H+ O(H(\log w)/w+n/w^c)+O(w^{4c})
    \end{align*}
    bits of space, where $H$ is defined as $\sum_{i=1}^{n}-\log D(a_i)$. The data structure also assumes access to a lookup table that depends only on $\Sigma$ and $D$, which uses $2^{\epsilon w}+O(|\Sigma|\cdot \log |\Sigma|)$ bits of space.
\end{theorem}

As discussed in Section \ref{sec:prelims}, we will often refer to the quantity $-\log D(\sigma)$ as the \defn{$D$-entropy} of a symbol $\sigma$ (it is always guaranteed that $D(\sigma)>0$).

\subsection{Tweaking the Distribution}

Before presenting the algorithm, we first modify the distribution $D$ so that no symbol is too costly. Currently, $-\log D(\sigma)$ could be much greater than $w$ for some $\sigma\in \Sigma$, which is very different from the prefix-code setting. To fix this, we tweak the distribution $D$ to $D'$, where
\begin{align*}
    D'(\sigma)=\frac{D(\sigma)+1/(|\Sigma|\cdot 2^w)}{1+1/2^w}
\end{align*}
for any $\sigma\in \Sigma$. It is easy to verify that $D'$ is a valid distribution and that $D'(\sigma)\ge\Omega(1/(|\Sigma|\cdot 2^w))$ for any $\sigma$, which means that $-\log D'(\sigma)=O(w)$ (recall that $|\Sigma|=2^{O(w)}$). Moreover, for any symbol $\sigma$,
\begin{align*}
    -\log D'(\sigma)\le -\log D(\sigma)+O(1/2^w),
\end{align*}
which means that if we replace $D$ with $D'$, the total $D$-entropy $H$ increases by at most $O(n/2^w)$, which is acceptable. Therefore, we may assume without loss of generality that $D'$ is used in place of $D$, and thus that $-\log D(\sigma) = O(w)$ for all $\sigma \in \Sigma$.

\subsection{Dealing with Long Blocks}

Recall the definition of blocks in \cref{thm:prefix-code-alg}: A block is an interval of entries that has total entropy $\Theta(w)$. If we directly use this requirement for the current setting, but with entropy replaced by $D$-entropy (that is, a block satisfies $\sum_{i}-\log D(a_i)=\Theta(w)$), then the blocks could contain much more than $w$ entries because the cost of storing a symbol $-\log D(a_i)$ can now be much smaller than $1$ bit.

To fix this, we additionally require that a block contains at most $\poly w$ entries. This requirement is the reason for the $n/\poly w$ term in the redundancy in Theorem \ref{thm:static-arithmetic-code-alg}. Note that if we were to remove this requirement, then even encoding the interval of entries that a block contains could require $\omega(\log w)$ bits, which would be a problem for the internal nodes in each B-tree---by restricting each block to $\poly(w)$ entries, we avoid this issue entirely. 

Concretely, in this section, we define chunks and blocks as follows. We partition the array into chunks of length $w^{3c}$ (rather than $w^{10}$ as in the previous section), where $c$ is the parameter given to us by Theorem \ref{thm:static-arithmetic-code-alg}. In defining the blocks, we require that any block must be completely contained in a chunk (this way, a block contains at most $w^{3c}$ entries). We assume that $n$ is at least $w^{3c}$ by adding dummy entries (which explains the $w^{4c}$ term in the redundancy). Within each chunk, we still use \cref{lem:B-tree} to maintain the blocks. Note that while the previous algorithm only considered chunks of length $w^{10}$, \cref{lem:B-tree} is actually much more general. Finally, concatenating the VMs for all chunks costs $O(n/w^{c})$ bits of space according to \cref{lem:polylog-concat}.

\subsection{Encoding the Block Contents}

We now describe how to encode the content of each block. 

\begin{lemma}
    \label{lem:encode_block}
    Given a block containing entries $[l,r]$, we can encode the symbols $a_l,\dots,a_r$ using
    \begin{align*}
        \bk*{\sum_{i\in [l,r]}-\log D(a_i)}+O(\log w)
    \end{align*}
    bits of space, while supporting efficient decodings and updates. This requires a lookup table of size $2^{\epsilon\cdot w}+O(|\Sigma|\cdot \log |\Sigma|)$.
\end{lemma}

\paragraph{Partitioning the block into intervals.} Let $\epsilon>0$ be a constant. We aim to partition the block into several intervals, where each interval represents a sequence of consecutive entries; and where each interval either has total $D$-entropy at most $(1/2)\epsilon\cdot w$, or contains only one symbol which has $D$-entropy larger than $(1/2)\epsilon\cdot w$. Formally, the intervals are determined as follows.

The first interval starts from the leftmost entry of the block. If the $D$-entropy of the first entry already exceeds $(1/2)\epsilon\cdot w$, then the first interval only contains one entry. Otherwise, we extend the first interval to the point where adding the next entry would increase the total $D$-entropy to more than $(1/2)\epsilon\cdot w$, or if it reaches the end of the block. We then do the same for the second interval (starting from where the first interval ended), and so on. 

\begin{claim}
    \label{clm:interval_number}
    The number of intervals in a block is $O(1/\epsilon)=O(1)$.
\end{claim}

\begin{proof}
    If any interval except the last is extended to the right by one entry, then its $D$-entropy exceeds $(1/2)\epsilon\cdot w$. Since the total $D$-entropy of these extended intervals is at most twice the total $D$-entropy of the block (each entry is counted at most twice), which is $O(w)$, we must have that the number of intervals is $O(1/\epsilon)=O(1)$.
\end{proof}

\paragraph{Encoding an interval.} When encoding a block, we use $O(\log w)$ bits to describe the partitioning of the intervals (i.e., to describe which slot each interval begins with), followed by the encoding of each interval (i.e., an encoding of what the entries are for the slots in the interval). Our basic approach to encoding intervals will be to use the Shannon encoding, which is a prefix code that incurs $O(1)$ bits of redundancy per symbol:

\begin{lemma}[Shannon Codes \cite{shannon1948mathematical}]
    \label{lem:encode}
    Let $\Sigma$ be an alphabet, and let $D$ be a distribution over $\Sigma$, such that $D(\sigma)>0$ for any $\sigma\in \Sigma$. There exists a prefix code encoding $\Sigma$, where the codeword for symbol $\sigma\in \Sigma$ has length $-\log D(\sigma)+O(1)$.
\end{lemma}

Critically, we will use Shannon codes not to encode the individual symbols in the array, but to encode entire intervals. Moreover, as we will see shortly, we actually make use of several different Shannon codes, with different codes used to encode different types of intervals. Specifically, we encode the intervals as follows:

\begin{itemize}
    \item For the intervals that contain only one symbol, we encode this symbol using the prefix code obtained from invoking \cref{lem:encode} on $D''$, defined as
    \begin{align*}
        D''(\sigma)=\frac{D(\sigma)+1/|\Sigma|}{2}.
    \end{align*}
    Under this distribution, the entropy of each symbol increases by at most $1$, since $-\log D''(\sigma)\le -\log D(\sigma)+1$. Moreover, each symbol $\sigma$ has probability at least $\Omega(1/|\Sigma|)$ in $D''$, which means that the length of its encoding in the prefix code is at most $\log (|\Sigma|)+O(1)$. Therefore, we can compute a lookup table of size $O(|\Sigma|)\cdot O(\log |\Sigma|)$ to map each codeword to the corresponding symbol.
    \item For intervals of length $\ell$ where $\ell>1$ and $\ell\le w^{3c}$, consider the distribution $D_\ell$ over $\Sigma^\ell$, where the $\ell$-tuple $(b_1,\dots,b_\ell)$ appears with probability $\prod_{i=1}^{\ell}D(b_i)$ (which is $>2^{-(1/2)\epsilon w}$ by our definition of intervals). We store the interval using the prefix code obtained by invoking \cref{lem:encode} on $D_\ell$. For each length $\ell$, we compute a lookup table to map codewords to symbols. Since each codeword appears with probability $>2^{-(1/2)\epsilon w}$, the length of the codewords corresponding to the symbols of an interval is at most $(1/2)\epsilon\cdot w+O(1)$, which means that the lookup tables use $2^{(1/2)\epsilon\cdot w}\cdot \ell\cdot O(\log|\Sigma|)\ll 2^{\epsilon w}$ bits of space in total.

    Note that some $\ell$-tuples may correspond to much longer codewords in the prefix code, but they can be ignored as they do not correspond to any valid interval.
\end{itemize}

In the above, we only describe lookup tables for decoding a block. We also need to use lookup tables for merging and splitting blocks as well as for updating symbols in blocks; these lookup tables are similar to the ones we have described. In total, the lookup tables use $2^{\epsilon\cdot w}+O(|\Sigma|\cdot \log |\Sigma|)$ bits of space. This concludes the proof of \cref{lem:encode_block}.

\paragraph{Putting the pieces together.}

To summarize, our data structure uses the two-level structure of \cref{thm:prefix-code-alg} to partition the array into chunks and blocks, then encodes the symbols within each block using Shannon codes, and maintains the encoding of blocks using B-trees. The overall space usage of our data structure is bounded as follows:
\begin{itemize}
    \item \textbf{Total encoding length of block:} For a block containing entries $[l,r]$, the length of our encoding is $\bk*{\sum_{i\in [l,r]}-\log D(a_i)}+O(\log w)$ using \cref{lem:encode_block}. Summing this up, the total encoding length is $H+O(\log w)\cdot (\text{\# of blocks})$. Since the number of blocks is $O(H/w)+O(n/w^{3c})$ (each block has total $D$-entropy $\Theta(w)$, except we allow one block of each of the $O(n/w^{3c})$ chunks to have small total entropy), the total encoding length is at most
    \begin{align*}
        H+O(H(\log w)/w+n/w^c).
    \end{align*}
    \item \textbf{Cost of B-trees:} We use \cref{lem:B-tree} to maintain the block contents, which, in addition to the lengths of the blocks' encodings, costs $O(\log w)$ extra bits per block and $O(w^{3c/2+2})$ extra bits per chunk. This sums up to (recall that we assumed $c\ge 4$)
    \begin{align*}
        O(H(\log w)/w+w^{3c/2+2}\cdot n/w^{3c})=O(H(\log w)/w+ n/w^{c}).
    \end{align*}
    \item \textbf{Total size of the lookup tables:} We need a lookup table of size $2^{\epsilon w}$ for processing the B-tree (see \cref{lem:B-tree}), and lookup tables of total size $2^{\epsilon\cdot w}+O(|\Sigma|\cdot \log |\Sigma|)$ for manipulating blocks' encodings (see \cref{lem:encode_block}).
\end{itemize}

In total, our data structure uses $H+O(H(\log w)/w+n/w^c)$ bits of space. Note that when $n< w^{3c}$, we add dummy entries to increase the array length to $w^{3c}$, which costs $O(w^{3c})$ extra bits. This concludes the construction of \cref{thm:static-arithmetic-code-alg}.

%% file: arithmetic-code.tex
\section{Dynamic Entropy-Encoded Arrays over a Sublinear Alphabet}
\label{sec:sublinear}

Next, we show how to construct a dynamic $O(1)$-time array that uses space very close to the entropy, so long as the alphabet size $|\Sigma|$ is slightly sublinear. 

\begin{theorem}
    \label{thm:arithmetic-code-alg}
    Let $c>1$ be a constant. Let $n$ be an integer. Let $w=\Omega(\log n)$ denote the machine-word size. Let $\Sigma=[|\Sigma|]$ be an alphabet, where $|\Sigma|=2^{O(w)}$. There exists a data structure in the virtual memory model that stores an array of symbols $a_1,\dots,a_n\in \Sigma$, and supports the following operations in constant time in the worst case:
    \begin{itemize}
        \item Update: Given $1\le i\le n$, change $a_i$ to some $s\in \Sigma$.
        \item Query: Given $1\le i\le n$, return $a_i$.
    \end{itemize}
    At any point in time, the data structure uses
    \begin{align*}
        H+O(H(\log \log n)/\log n+n/(\log n)^c+|\Sigma|\cdot w^2)
    \end{align*}bits of space. The entropy $H$ is defined as
    \begin{align*}
        H=\sum_{i=1}^{n}-\log(f_{a_i}/n)=\sum_{f_\sigma\ne 0}f_\sigma\cdot \log(n/f_\sigma),
    \end{align*} where $f_\sigma$ is the \emph{current} frequency of the symbol $\sigma$ (i.e., $H$ changes over time).
\end{theorem}

We define the following notations for clarity of discussion: After the $t$-th update, let $a_i^{(t)}$ denote the current value of the $i$-th entry, let $f_{\sigma}^{(t)}$ denote the current frequency of $\sigma$, and let $D^{(t)}$ denote the current empirical distribution. Finally, we will also use $H^{(t)}$ to denote the current entropy, which is
\begin{align*}
    H^{(t)}=\sum_{i=1}^{n}-\log D^{(t)}(a_i^{(t)}).
\end{align*}

\paragraph{WLOG assumptions on parameters.}
To simplify our exposition, throughout the section, we make the following assumptions on the parameters $c,H$ and $|\Sigma|$, each of which holds without loss of generality:
\begin{itemize}
    \item $c$ is a sufficiently large positive constant: This is WLOG because, replacing $c$ with a larger constant only makes the theorem harder to prove.
    \item $|\Sigma|=O(n)$: This is WLOG because, if it does not hold, then we can trivially achieve Theorem \ref{thm:arithmetic-code-alg} by storing the array explicitly, using $O(n\log |\Sigma|)=o(|\Sigma|\cdot w^2)$ bits of space.
    \item $|\Sigma|=\Omega(n/(\log n)^{3c})$: This is WLOG because, if it does not hold, we can increase $|\Sigma|$ without changing the space guarantee that we are aiming for.
    \item $H = \Omega(|\Sigma| \cdot \log^2 n)$, which in turn implies that $H = \omega(n/(\log n)^{3c})$: This is WLOG because our space bound in \cref{thm:arithmetic-code-alg} is at least $\Omega(|\Sigma| \cdot w^2) = \Omega(|\Sigma| \log^2 n)$. If $H$ is smaller than this space bound, we can add $|\Sigma|$ dummy symbols to $\Sigma$, and append $\log n$ copies of each dummy symbol to the array, which are never updated. These dummy symbols only change the desired space bound by $O(|\Sigma| \log^2 n)$ bits, which can be absorbed into the $O(|\Sigma| \cdot w^2)$ term in \cref{thm:arithmetic-code-alg}.
\end{itemize}

With these assumptions in place, we can now present the construction used to prove Theorem \ref{thm:arithmetic-code-alg}.

\subsection{\texorpdfstring{Supporting $O(H^{(0)} / \log n)$ Updates}{Supporting O(H(0) / log n) Updates}}

When the empirical distribution is fixed, the data structure of \cref{thm:static-arithmetic-code-alg} suffices for \cref{thm:arithmetic-code-alg}. The main challenge is therefore to generalize \cref{thm:static-arithmetic-code-alg} to the case where the empirical distribution can change over time.

In the algorithm, we initially encode the array under $\widetilde{D}^{(0)}$, which is a modified version of $D^{(0)}$ (defined later). We will show that, after a small number of updates, the array can still be efficiently encoded under $\widetilde{D}^{(0)}$. Then, after processing some updates, we rebuild the data structure and re-encode the array under $\widetilde{D}^{(t)}$ (where $t$ is the time at which we perform the rebuild), which is a modified version of $D^{(t)}$.

In this section, we assume that the number of updates is only $O(H^{(0)} / \log n)$, in which case there is no need to rebuild, and we show how to support these $O(H^{(0)} / \log n)$ updates while keeping the data structure space-efficient.

\begin{lemma}
    \label{lem:arith-no-rebuild}
    There exists a data structure that satisfies the requirements of \cref{thm:arithmetic-code-alg}, except that it only supports $c_{\text{rebuild}}\cdot H^{{(0)}}/\log n$ updates for some sufficiently small constant $c_{\text{rebuild}}$.
\end{lemma}

To get some intuition for this case, consider the straightforward approach where we simply use \cref{thm:static-arithmetic-code-alg} to encode the symbols under the initial empirical distribution $D^{(0)}$. This approach works if the empirical distribution is fixed, but quickly runs into problems if the frequency of some symbol increases by too much: In $t$ updates, the frequency of a symbol $\sigma$ could increase from $1$ to $t$, which means that each occurrence of $\sigma$ costs $\log t$ bits more than optimal if encoded under $D^{(0)}$ instead of $D^{(t)}$. If, for example, $H^{(0)}=\Theta(n\log n)$ and $t=\Omega(H^{(0)}/\log n)$, then our encoding length becomes longer than $H^{(t)}$ by $\Theta(n\log n)$, which is unacceptable.

To handle this problem, we encode the array under a different distribution $\widetilde{D}^{(0)}$, which is obtained from $D^{(0)}$ by adding placeholders. These placeholders are used to encode symbols that were rare in $D^{(0)}$, but became much more frequent after updates.

\paragraph{Introducing placeholder distributions.} Given a distribution $D$, its \defn{placeholder distribution} $\widetilde{D}$ is defined as follows:
\begin{itemize}
    \item For each symbol $\sigma\in \Sigma$, we have that $\widetilde{D}(\sigma)=D(\sigma)\cdot (1-(\log n)^{-10c})$.
    \item For $0\le k\le \lfloor\log n\rfloor+1$, we add $\min(|\Sigma|,2^{k+1})$ extra symbols $\tau_{k,1},\tau_{k,2},\dots$, each having probability
    \begin{align*}
        \frac 1{\min(|\Sigma|,2^{k+1})\cdot (\log n)^{10c}\cdot(\lfloor\log n\rfloor+2)}
    \end{align*}
    in $\widetilde{D}$. These are called the \defn{level-$k$ placeholders}.

    We can verify that $\widetilde{D}$ is a proper distribution: For each $0\le k\le \lfloor\log n\rfloor+1$, the total probability of level-$k$ placeholders is equal to $1/((\log n)^{10c}\cdot(\lfloor\log n\rfloor+2))$. Since there are $\lfloor\log n\rfloor+2$ levels, the total probability of all placeholders is $1/(\log n)^{10c}$. As for the original symbols, we have $\sum_{\sigma\in \Sigma}\widetilde{D}(\sigma)=\sum_{\sigma\in \Sigma}{D}(\sigma)\cdot (1-(\log n)^{-10c})=(1-(\log n)^{-10c})$, so the total probability sums to $1$.
\end{itemize}

Intuitively, a level-$k$ placeholder is used to represent a symbol that appears with probability $2^{-k}$.

\paragraph{Handling updates using placeholders.} We encode the array under the placeholder distribution $\widetilde{D}^{(0)}$ of $D^{(0)}$ (until the first rebuild), where the placeholders are assigned to symbols in $\Sigma$, and each symbol in $\Sigma$ is assigned at most one placeholder of each level. This assignment is initially empty, and is determined on the fly. Initially, each entry $a_i=\sigma$ is encoded using the original symbol $\sigma$. When we perform an update, the modified entry will be encoded using a placeholder.

More specifically, the updates are handled as follows: Suppose that we are currently performing the $t$-th update, and assume that the update changes the $i$-th entry to $\sigma$. We use $O(|\Sigma|\cdot \log n)$ bits of extra space to record the true frequencies $f_\sigma^{(t)}$. Let $k$ be the smallest integer such that $2^{-k}< D^{(t)}(\sigma)=f_{\sigma}^{(t)}/n$. We encode the $i$-th entry using the level-$k$ placeholder assigned to $\sigma$. If $\sigma$ does not own one, then we assign the smallest unassigned level-$k$ placeholder to $\sigma$. We show in the following that such a placeholder must exist.

\paragraph{Analysis.} First, we show that when the number of updates is small, we will not run out of placeholders.
\begin{claim}
    \label{clm:sufficient_placeholders}
    Within the first $n$ updates, for any $k$, the number of symbols $\sigma\in \Sigma$ that have requested a level-$k$ placeholder does not exceed $2^{k+1}$.
\end{claim}
\begin{proof}
    Let $g_\sigma$ denote the total number of updates that change some symbol to $\sigma$, among the first $n$ updates. We have that $\sum_{\sigma\in \Sigma} g_{\sigma}=n$.
    
    When a symbol $\sigma$ requests a level-$k$ placeholder at time $t$, we have that $2^{-k}<f^{(t)}_{\sigma}/n$. Note that $f^{(t)}_{\sigma}\le f^{(0)}_{\sigma}+g_{\sigma}$ by definition of $g_{\sigma}$, which implies that $f^{(0)}_{\sigma}+g_\sigma>n/2^k$. Letting $T\subseteq \Sigma$ denote the set of symbols that have ever requested a level-$k$ placeholder, we have that
    \begin{align*}
        |T|\cdot n/2^k<\sum_{\sigma\in T}f^{(0)}_{\sigma}+g_\sigma\le \sum_{\sigma\in \Sigma}f^{(0)}_{\sigma}+g_\sigma\le 2n.
    \end{align*}
    Therefore, $|T|<2^{k+1}$.
\end{proof}

Next, we show that the encoding under the placeholder distribution is not too costly compared to the optimal encoding. Let $\widetilde{a}_i^{(t)}$ denote the encoding of the $i$-th entry under $\widetilde{D}^{(0)}$, as described before (that is, $\widetilde{a}_i^{(t)}$ is either $a_i^{(t)}$ or a placeholder assigned to $a_i^{(t)}$). Let $\widetilde{H}^{(t)}$ denote the cost of storing the array under the placeholder distribution $\widetilde{D}^{(0)}$, defined as
\begin{align*}
    \widetilde{H}^{(t)}=\sum_{i=1}^{n}-\log \widetilde{D}^{(0)}(\widetilde{a}_i^{(t)}).
\end{align*}
Note that before the first rebuild, $\widetilde{H}^{(t)}$ is always defined with respect to $\widetilde{D}^{(0)}$, not $D^{(t)}$.

Let $\Delta^{(t)}$ be a measure of difference between our encoding and the optimal encoding, defined as
\begin{align*}
    \Delta^{(t)}=\sum_{i=1}^{n}\max\BK*{0,-\log \widetilde{D}^{(0)}(\widetilde{a}_i^{(t)})+\log D^{(t)}(a_i^{(t)})}.
\end{align*}
Note that $\Delta^{(t)}$ is not the same as $\widetilde{H}^{(t)}-H^{(t)}$ (because we do not allow for negative summands in $\Delta^{(t)}$), but it does satisfy $\Delta^{(t)} \ge \widetilde{H}^{(t)}-H^{(t)}$. The reason that we define $\Delta^{(t)}$ in this way is so that, for $t \in O(n)$, we always have $\Delta^{(t)}-\Delta^{(t-1)}=O(\log \log n)$ (we will prove this in a moment). In comparison, if we simply define $\Delta^{(t)}$ as $\widetilde{H}^{(t)}-{H}^{(t)}$, then $\Delta^{(t)}-\Delta^{(t-1)}$ could be as large as $O(\log n)$.

The following claim establishes an upper bound on how large $\Delta^{(t)}$ can be.

\begin{claim}
    \label{clm:delta_decrement_slow}
    We have that $\Delta^{(t)}=O(n/(\log n)^{10c}+t\log \log n)$.
\end{claim}

\begin{proof}
    We prove the claim by induction. Initially, $\Delta^{(0)}=O(n/(\log n)^{10c})$: For any symbol $\sigma\in \Sigma$, its probability $\widetilde{D}^{(0)}(\sigma)$ in the placeholder distribution is smaller than $D^{(0)}(\sigma)$ by a $1-(\log n)^{-10c}$ factor. Thus, the encoding length under $\widetilde{D}^{(0)}$ is longer than that under $D^{(0)}$ by $-\log(1-(\log n)^{-10c})=O((\log n)^{-10c})$ for every entry.

    Next, we show that $\Delta^{(t)}-\Delta^{(t-1)}=O(\log \log n)$. Suppose that the $t$-th update changes the $i$-th entry to $\sigma$. We upper bound $\Delta^{(t)}-\Delta^{(t-1)}$ by comparing the summands:
    \begin{itemize}
        \item \textbf{Comparing the $i$-th term.} We will argue that the $i$-th summand of $\Delta^{(t)}$ is at most $O(\log \log n)$ larger than the $i$-th summand in $\Delta^{(t - 1)}$. Recall that $\widetilde{a}_i^{(t)}$ is a level-$k$ placeholder, where $k$ is the smallest integer such that $2^{-k}<f_\sigma^{(t)}/n$ (that is, $2^{-(k-1)}\ge f_\sigma^{(t)}/n$). By definition, we have that
        \begin{align*}
            \widetilde{D}^{(0)}(\widetilde{a}_i^{(t)})\ge \frac{1}{\min(|
            \Sigma|,2^{k+1})\cdot (\log n)^{10c+1}\cdot(\lfloor\log n\rfloor+2)}.
        \end{align*}
        Therefore, the encoding length of $\widetilde{a}_i^{(t)}$ is at most
        \begin{align*}
            -\log \widetilde{D}^{(0)}(\widetilde{a}_i^{(t)})&\le \log (\min(|
            \Sigma|,2^{k+1})\cdot (\log n)^{10c+1}\cdot(\lfloor\log n\rfloor+2)) \\
            &\le -\log(f_\sigma^{(t)}/(n\cdot \poly \log n))= -\log D^{(t)}(\sigma)+O(\log \log n).
        \end{align*}
        Thus the $i$-th summand in $\Delta^{(t)}$ is at most $O(\log \log n)$. In comparison, the $i$-th summand in $\Delta^{(t-1)}$ is at least $0$, so the increase in the $i$-th summand between $\Delta^{(t-1)}$ and $\Delta^{(t)}$ is at most $O(\log \log n)$.
        \item \textbf{Comparing the other terms.} We show that the other terms in the sum collectively increase by $O(1)$ when switching from $\Delta^{(t-1)}$ to $\Delta^{(t)}$.
        
        For the terms excluding $i$, their encoding length under the placeholder distribution does not change at time $t$, so any increase is due to the change in the empirical distribution $D^{(t)}$. That is, the total increase to $\Delta^{(t)}$ from these terms can be bounded as
        \begin{align*}
            &\sum_{i'\ne i}\max\BK*{0,-\log \widetilde{D}^{(0)}(\widetilde{a}_{i'}^{(t)})+\log D^{(t)}(a_{i'}^{(t)})}-\max\BK*{0,-\log \widetilde{D}^{(0)}(\widetilde{a}_{i'}^{(t-1)})+\log D^{(t-1)}(a_{i'}^{(t-1)})} \\
            \le{}&\sum_{i'\ne i}\max\BK*{0,\log D^{(t)}(a_{i'}^{(t)})-\log D^{(t-1)}(a_{i'}^{(t)})}.\tag{$\widetilde{a}_{i'}^{(t-1)}=\widetilde{a}_{i'}^{(t)}$ for $i'\ne i$}
        \end{align*}
        We assume that $a_{i}^{(t-1)}\ne \sigma$ (i.e., the $t$-th update is not trivial), since otherwise the $t$-th update does not change $D^{(t)}$. 
        
        In this case, if $D^{(t)}(a_{i'}^{(t)})$ is larger than $D^{(t-1)}(a_{i'}^{(t)})$, then we must have $a_{i'}^{(t)}=\sigma$. Therefore, the above summation is at most
        \begin{align*}
            &\sum_{i'\ne i,a_{i'}^{(t)}=\sigma}\log D^{(t)}(a_{i'}^{(t)})-\log D^{(t-1)}(a_{i'}^{(t)}) \\
            ={}&f_\sigma^{(t-1)}\cdot (\log D^{(t)}(\sigma)-\log D^{(t-1)}(\sigma)) \\
            ={}&f_\sigma^{(t-1)}\cdot (\log (f_\sigma^{(t-1)}+1)-\log f_\sigma^{(t-1)})\tag{by definition of $D^{(t)}$} \\
            <{}&f_\sigma^{(t-1)}\cdot \frac{\log e}{f_\sigma^{(t-1)}}=O(1).\tag{$\ln(1+x)<x$}
        \end{align*}
    \end{itemize}
    This proves that $\Delta^{(t)} - \Delta^{(t-1)} = O(\log \log n)$, which implies the claim.
\end{proof}

We now conclude this section.

\begin{proof}[Proof of \cref{lem:arith-no-rebuild}]
    We encode the array under the initial placeholder distribution $\widetilde{D}^{(0)}$ using \cref{thm:static-arithmetic-code-alg}, and maintain the correspondence between placeholders and symbols in $\Sigma$ using a table of $O(|\Sigma|\cdot \log^2 n)$ bits.
    
    \cref{clm:delta_decrement_slow} implies that, at time $t\le c_{\text{rebuild}}\cdot H^{(0)}/\log n$, our encoding length is at most
    \begin{align*}
        H^{(t)}+O(n/(\log n)^{10c}+H^{(0)}\log \log n/\log n)
    \end{align*}
    bits. Each update changes the entropy $H^{(t)}$ by at most $O(\log n)$ (the proof is the same as that of \cref{clm:delta_decrement_slow}). Thus, when $t\le c_{\text{rebuild}}\cdot H^{(0)}/\log n$ for some sufficiently small constant $c_{\text{rebuild}}$, we have $H^{(t)}\in H^{(0)}\cdot [0.9,1.1]$. It follows that our encoding length is at most
    \begin{align*}
        H^{(t)}+O(n/(\log n)^{10c}+H^{(t)}\log \log n/\log n)
    \end{align*}
    bits. Feeding this into \cref{thm:static-arithmetic-code-alg} gives us a data structure that uses
    \begin{align*}
        H^{(t)}+O(H^{(t)}(\log\log n)/\log n+n/\log^c n+|\Sigma|\cdot \log^2 n)
    \end{align*}
    bits of space whenever $t\le c_{\text{rebuild}}\cdot H^{(0)}/\log n$. 
    
    Note that while the actual word size may be larger than $\log n$, here we only invoke \cref{thm:static-arithmetic-code-alg} with word size $w=\Theta(\log n)$, which guarantees that our lookup tables have size $n^{1-\Omega(1)}$. Also note that the space usage has an additive $O(|\Sigma|\cdot \log ^2n)$ term instead of $O(|\Sigma|\cdot w)$ as in \cref{thm:static-arithmetic-code-alg}, which is due to the fact that the alphabet size of the placeholder distribution $\widetilde{D}^{(0)}$ is  $O(|\Sigma|\cdot \log n)$.
\end{proof}

\subsection{The Amortized Algorithm}

In this section, we generalize \cref{lem:arith-no-rebuild} to allow for any number of updates. We achieve this by rebuilding the array at certain points in time, where we re-encode the entire array under the most recent placeholder distribution $\widetilde{D}^{(t)}$ ($t$ is the time of the rebuild). We will do so in a way that rebuilds cost constant amortized time per update. Later, we will show how to deamortize the rebuilds.

\begin{lemma}
    \label{lem:arith-amortized}
    There exists a data structure that satisfies the requirements of \cref{thm:arithmetic-code-alg}, except that the updates are handled in amortized constant time.
\end{lemma}

\paragraph{Criterion for initiating a rebuild.} Let $c_{\text{rebuild}}>0$ be the constant in \cref{lem:arith-no-rebuild}. Define $t_0=0$. For $k\ge 1$, the $k$-th rebuild is scheduled at time $t_k=t_{k-1}+\lceil c_{\text{rebuild}}\cdot H^{(t_{k-1})}/\log n\rceil$. That is, the time for the $k$-th rebuild can be determined right after performing the $(k-1)$-th rebuild.

To know when to rebuild, we maintain the total entropy $H^{(t)}$, and update $H^{(t)}$ whenever the frequency $f_\sigma^{(t)}$ of some symbol changes (recall that we already store the frequencies $f_\sigma^{(t)}$ in the data structure for knowing which level of placeholder to use). Note that $H^{(t)}=O(n\cdot \log n)$, so when $c_{\text{rebuild}}$ is a sufficiently small constant, we will have that $t_k-t_{k-1}\le n/4$ for any $k$, in which case we can apply \cref{clm:sufficient_placeholders} to show that we never run out of placeholders.

As discussed in the proof of \cref{lem:arith-no-rebuild}, each update changes $H^{(t)}$ by at most $O(\log n)$, so when $c_{\text{rebuild}}$ is a sufficiently small constant, for any $k$, we have $H^{(t)}\in [0.9,1.1]H^{(t_{k-1})}$ while $t\in [t_{k-1},t_k]$.

We will carefully implement the rebuild, so that the $k$-th rebuild takes $O(H^{(t_{k-1})}/\log n)$ time to complete. This is possible since the number of blocks is $O(H^{(t_{k-1})}/\log n)$ (because of $H^{(t)}\in [0.9,1.1]H^{(t_{k-1})}$ and \cref{clm:delta_decrement_slow}), and we can rebuild one block in $O(1)$ time. Since the number of updates between the $(k-1)$-th and the $k$-th rebuild is $\Theta(H^{(t_{k-1})}/\log n)$, this guarantees that the amortized time cost of the rebuilds is $O(1)$.

Our definition of rebuild also ensures that the data structure is space-efficient: For the times $t\in [t_0,t_1)$ before the first rebuild, we can use \cref{clm:delta_decrement_slow} to show that
\begin{align*}
    \widetilde{H}^{(t)}-H^{(t)}&\le \Delta^{(t)} \\
    &\le O(n/(\log n)^{10c}+t_1\log \log n)\tag{for any $t\le t_1$} \\
    &=O(H^{(0)}\log \log n/\log n),
\end{align*}
so our encoding length is always efficient. The general case is the same.

\paragraph{Implementing a rebuild.} In this paragraph, we only describe the first rebuild, since the others are similar. When performing the first rebuild at time $t_1$, we will re-encode the array using $O(1)$ time per block. Before the rebuild, the number of blocks is $O(\widetilde{H}^{(t_1)}/\log n)$, which is $O({H}^{(0)}/\log n)$ due to \cref{clm:delta_decrement_slow} and the fact that $c_{\text{rebuild}}$ is small. Therefore, the time cost of the first rebuild is $O(H^{(0)}/\log n)$ instead of the na\"{\i}ve time bound of $O(n)$. The rebuild uses a lookup table of size $O(|\Sigma|\cdot \log^2 n)+2^{\epsilon \log n}$ (which is computed on the fly), and is performed as follows:
\begin{enumerate}
    \item First, we decompose the blocks in the array into intervals. Recall that, when encoding a block in \cref{thm:static-arithmetic-code-alg}, blocks are decomposed into intervals of cells, such that every interval either has total entropy at most $(1/2)\epsilon\cdot \log n$ (we use $w=\log n$ when invoking \cref{thm:static-arithmetic-code-alg}), or is simply a singleton.

    We now perform the exact same decomposition to obtain an array of intervals. It is shown in \cref{clm:interval_number} that the number of intervals in a block is $O(1)$, so the total number of intervals is $O(H^{(0)}/\log n)$.

    \textbf{Lookup tables:} For this, we only need the lookup tables in \cref{thm:static-arithmetic-code-alg} for decoding the block contents. 
    \item Next, we re-encode each interval under the new placeholder distribution $\widetilde{D}^{(t_1)}$, where every old interval under $\widetilde{D}^{(0)}$ is converted to a sequence of new intervals under $\widetilde{D}^{(t_1)}$. Note that one old interval may be converted to a sequence of multiple new intervals, because the total entropy of the interval may increase. However, using the same argument as \cref{clm:interval_number}, we have that each old interval is converted to at most $O(H'/\log n+1)$ intervals, where $H'$ is the total $\widetilde{D}^{(t_1)}$-entropy of the interval. Thus, the total number of new intervals is at most
    \begin{align*}
        O(\text{\# of old intervals})+O(H^{(t_1)}/\log n)=O(H^{(0)}/\log n).
    \end{align*}
    
    \textbf{Lookup tables:} Note that the number of possible intervals under $\widetilde{D}^{(0)}$ is $O(|\Sigma|\cdot \log n)$ (the singleton intervals) plus $2^{(1/2)\epsilon\cdot \log n}\cdot (\log n)^{3c}$ (we have at most $2^{(1/2)\epsilon\cdot \log n}$ possible intervals of each length). Therefore, to re-encode the intervals, we only have to compute a lookup table of size $O(|\Sigma|\cdot \log^2 n)+2^{\epsilon \log n}$ bits.
    \item Finally, we merge intervals into blocks. That is, we first think of each interval as a block, and then merge adjacent blocks if their total entropy does not exceed $O(\log n)$. 
    
    \textbf{Lookup tables:} For merging blocks, we can use the same lookup tables as the ones in \cref{thm:static-arithmetic-code-alg} (when invoked for $\widetilde{D}^{(t_1)}$).
\end{enumerate}

Aside from rebuilds, the other parts of our amortized algorithm are the same as \cref{lem:arith-no-rebuild}. This concludes the proof of \cref{lem:arith-amortized}.

\subsection{Deamortization}

In this section, we show how to improve the update time to worst-case constant. We only consider what happens \emph{between} the first and second rebuilds, since the other cases are similar.

\paragraph{Intuition.} Suppose that the first rebuild happens right after time $t_1$, and the second right after time $t_2$. In order to deamortize, instead of fully executing the first rebuild at time $t_1$, we slowly run the first rebuild in the background until time $t_2$, in a way that achieves the following guarantees:
\begin{itemize}
    \item we use $O(1)$ time per operation to perform the rebuild;
    \item the first rebuild is fully executed before the second rebuild is initiated;
    \item at any point in time, the data structure is space-efficient.
\end{itemize}

\paragraph{Notation.} To avoid confusion about which placeholder distribution we are discussing, we define the following notation: Let $\widetilde{a}^{(t,1)}_i$ $(t\in [t_1,t_2])$ denote the symbol used to encode $a_i^{(t)}$ under the placeholder distribution $\widetilde{D}^{(t_1)}$ (which is exactly the same as in the amortized algorithm), and let $\widetilde{H}^{(t,1)}$ $(t\in [t_1,t_2])$ denote the encoding length of the symbols $\widetilde{a}^{(t,1)}_i$ under $\widetilde{D}^{(t_1)}$.
Similarly define $\widetilde{a}^{(t,0)}_i$: $\widetilde{a}^{(t,0)}_i$ $(t\in [t_1,t_2])$ is the symbol that would be used to encode $a_i^{(t)}$ under $\widetilde{D}^{(0)}$, had the first rebuild never taken place.

\paragraph{Running rebuilds in the background.} Recall that in the amortized algorithm, the second rebuild was initiated when $(t_2-t_1)\ge c_{\text{rebuild}}\cdot H^{(t_1)}/\log n$. We use the same criterion in the new algorithm. However, the difference between the two algorithms is that, during $[t_1,t_2]$, not all symbols are encoded under $\widetilde{D}^{(t_1)}$: Some still use the old distribution $\widetilde{D}^{(0)}$.

We first compute the necessary lookup tables, then rebuild the array one chunk at a time. At any point in time, exactly one chunk is undergoing rebuild, while the other chunks are either entirely encoded under $\widetilde{D}^{(0)}$ or under $\widetilde{D}^{(t_1)}$.

Formally, the subroutine for rebuilding one chunk is as follows. Keep in mind that in the real algorithm, this subroutine is executed in the background, where it is given $O(1)$ time per update.

\begin{enumerate}
    \item Initialize a B-tree using \cref{lem:B-tree} that is used to temporarily store the updates to the current chunk during the chunk's rebuild. This takes $O(1)$ time. 
    \item Let $t'$ be the current time. We perform the rebuild which changes the encoding of every entry in the chunk from $\widetilde{a}_i^{(t',0)}$ to $\widetilde{a}_i^{(t',1)}$. The rebuild is done block by block.
    
    During this step, the updates are handled as follows: If the updated entry is not in the current chunk, it is encoded like before. Otherwise, we insert the update into the B-tree. When answering a query to the current chunk, we first check if the queried entry is in the B-tree. If not, we query the VM of the current chunk.
    \item Finally, we perform all updates temporarily stored in the B-tree, and delete the B-tree. During this step, the incoming updates are directly implemented, with one difference from the normal situation: If the updated entry is also present in the B-tree, then we need to remove it from the B-tree, so that we don't rewrite new updates with old ones.
\end{enumerate}

We can bound the total time cost of rebuilding a chunk as follows. Notice that the time spent in Step 3 is (up to constant factors) at most the time spent in Step 2, since each update implemented in Step 3 had to arrive during Step 2. Thus, it suffices to bound the time spent in Step 2, which is $O(\widetilde{H}'/\log n)+O(1)$, where 
\begin{align*}
    \widetilde{H}'=\sum_{i\text{ in current chunk}}-\log \widetilde{D}^{(0)}(\widetilde{a}^{(t',0)}_i)-\log \widetilde{D}^{(t_1)}(\widetilde{a}^{(t',1)}_i),
\end{align*}
and where $t'$ is the starting time of Step 2. 

\paragraph{Analysis of time.} We show that if the rebuild is allowed to run for $c_{\text{background}}$ steps per update for some sufficiently large constant $c_{\text{background}}$ (depending on $c_{\text{rebuild}}$ and $c$), then the first rebuild is always finished before the second rebuild starts.

Let $\widetilde{H}_{\text{max}}$ denote an upper bound on the encoding length during $[t_1,t_2]$:
\begin{align*}
    \widetilde{H}_{\max}=\sum_{i=1}^{n}\max_{t\in [t_1,t_2]}\BK*{-\log \widetilde{D}^{(0)}(\widetilde{a}^{(t,0)}_i),-\log \widetilde{D}^{(t_1)}(\widetilde{a}^{(t,1)}_i)}.
\end{align*}

Comparing this with $\widetilde{H}'$ which is the time cost of rebuilding one chunk, we see that the total time cost of the first rebuild is at most (we can let the number of chunks be much smaller than $n/(\log n)^{3c}$)
\begin{align*}
    O\bk*{\widetilde{H}_{\text{max}}/\log n}+o(n/(\log n)^{3c}).
\end{align*}

Thus, for the rebuilds to have enough time, it suffices to show that
\begin{align}
    \label{equ:deamortize_goal}
    t_2-t_1=\Omega(\widetilde{H}_{\text{max}}/\log n)
\end{align}
Recall that $H^{(t)}=\Omega(n/(\log n)^{3c})$ always holds, so the same bound holds for $\widetilde{H}_{\text{max}}$. For this we use the following claim:

\begin{claim}
    \label{clm:delta_decrement_slow_max}
    
    The upper bound $\widetilde{H}_{\text{max}}$ is at most $O(H^{(t_1)}+(t_2-t_1)\cdot \log \log n)$.
\end{claim}

\begin{proof}
    Similar to \cref{clm:delta_decrement_slow}, we define an inductive goal as follows: For $t\in [t_1,t_2]$, let
    \begin{align*}
        \widetilde{H}_{\text{ind}}^{(t)}=\sum_{i=1}^{n}\max_{t'\in [t_1,t]}\BK*{-\log \widetilde{D}^{(0)}(\widetilde{a}^{(t',0)}_i),-\log \widetilde{D}^{(t_1)}(\widetilde{a}^{(t',1)}_i)},
    \end{align*}
    and let
    \begin{align*}
        H_{\text{ind}}^{(t)}=\sum_{i=1}^{n}\max_{t'\in [t_1,t]}\BK*{-\log D^{(t')}(a_i^{(t')})}.
    \end{align*}
    The same inductive argument as in \cref{clm:delta_decrement_slow} shows that
    \begin{align*}
        \widetilde{H}_{\text{ind}}^{(t)}-H_{\text{ind}}^{(t)}=O(n/(\log n)^{10c}+H^{(t_1)}\log\log n/\log n+(t-t_1)\log\log n)
    \end{align*}
    and that $H_{\text{ind}}^{(t)}=O(H^{(t_1)})$ for any $t\in [t_1,t_2]$. The details are omitted.
\end{proof}

Comparing \cref{clm:delta_decrement_slow_max} with our goal \eqref{equ:deamortize_goal}, it remains to show that $t_2-t_1=\Omega(H^{(t_1)}/\log n)$. This holds because of our criterion for starting the second rebuild.

\paragraph{Analysis of space.} 

Let $\widetilde{H}_{\text{max}}^{(t)}$ denote an upper bound on the encoding length \emph{at} time $t\in [t_1,t_2]$:
\begin{align*}
    \widetilde{H}_{\max}^{(t)}=\sum_{i=1}^{n}\max\BK*{-\log \widetilde{D}^{(0)}(\widetilde{a}^{(t,0)}_i),-\log \widetilde{D}^{(t_1)}(\widetilde{a}^{(t,1)}_i)}.
\end{align*}

Again, using the inductive arguments of \cref{clm:delta_decrement_slow}, we can show the following which implies space efficiency of the data structure at time $t$. The proof is omitted.

\begin{claim}
    \label{clm:delta_decrement_slow_max_t}
    
    For $t\in [t_1,t_2]$, we have that $\widetilde{H}_{\text{max}}^{(t)}-H^{(t)}\le O(n/(\log n)^{10c}+H^{(t_1)}\log\log n/\log n+(t-t_1)\log\log n)$.
\end{claim}

In summary, we have shown that our rebuilds run in worst-case constant time and are space-efficient. 
Using this on top of the previous amortized algorithm \cref{lem:arith-amortized} concludes the proof of \cref{thm:arithmetic-code-alg}.

%% file: arithmetic-code-large-alphabet.tex
\section{Dynamic Entropy-Encoded Arrays over an Arbitrary Alphabet}
\label{sec:fullresult}
In this section, we further improve \cref{thm:arithmetic-code-alg} to remove the dependency of $|\Sigma|$ in the space usage, so that the data structure can be used for alphabets of size up to $2^{O(w)}$ (i.e., arbitrarily large alphabets, so long as the symbols fit in a machine word).

\subsection{Step 1: Alphabets of Size \texorpdfstring{$n \poly\log n$}{n poly log n}}

We begin by showing how to support alphabets of size $n \polylog n$.
\begin{theorem}
    \label{thm:arithmetic-code-n-poly-log-n}
    Let $c>1$ be a constant. When $|\Sigma|=n\poly\log n$, there exists a data structure that achieves the same guarantees as in \cref{thm:arithmetic-code-alg} and uses only
    \begin{align*}
        H+O(H(\log \log n)/\log n+n/(\log n)^c)
    \end{align*}
    bits of space. However, both the time and space bounds are with high probability in $n$.
\end{theorem}

In the data structure of \cref{thm:arithmetic-code-alg}, the space dependency on $|\Sigma|$ comes from two parts: Storing the frequencies $f_\sigma^{(t)}$, and storing the lookup tables used to decode the blocks. In the following, we improve the space usage of both parts.

\subsubsection{Space-efficient Storage of Frequencies}

Let $m^{(t)}$ denote the number of distinct symbols in the array at time $t$: $m^{(t)}=\sum_\sigma[f_\sigma^{(t)}\ne 0]$. The improvement in storing the frequencies is based on the observation that the data structure is allowed to waste $O(m^{(t)}\cdot \log \log n)$ bits of space.

\begin{claim}
    \label{clm:distinct_symbols_cheap}
    $m^{(t)}\cdot \log \log n=O(H^{(t)}(\log\log n)/\log n+\log\log n)$.
\end{claim}

\begin{proof}
    We can lower bound $H^{(t)}$ as
    \begin{align*}
        H^{(t)}&=\sum_{f_\sigma^{(t)}\ne 0}f_\sigma^{(t)}\cdot \log(n/f_\sigma^{(t)}) \\
        &\ge \sum_{0<f_\sigma^{(t)}\le n/2}\Omega(\log n) \tag{$f_\sigma^{(t)}\cdot \log(n/f_\sigma^{(t)})=\Omega(\log n)$ when $0<f_\sigma^{(t)}\le n/2$} \\
        &=\Omega((m^{(t)}-1)\cdot \log n).
    \end{align*}
    In other words, each symbol with frequency at most $n/2$ contributes at least $\Omega(\log n)$ to the entropy $H^{(t)}$. This implies that $m^{(t)}\cdot \log \log n=O(H^{(t)}(\log \log n)/\log n+\log\log n)$, where the extra $O(\log\log n)$ term is the cost of the most frequent symbol.
\end{proof}

\paragraph{Two-level storage.} Based on the above fact, we use two resizable dictionaries \cite{bender2022optimal} to store the frequencies.
\begin{itemize}
    \item The level-$1$ dictionary stores all $m^{(t)}$ symbols $\sigma$ with $f_{\sigma}^{(t)}\ne 0$. Each symbol is associated with an $O(\log \log n)$-bit value, which is equal to $\min(f_\sigma^{(t)},\lceil\log n\rceil)$. The space usage is
    \begin{align*}
        \log\binom{|\Sigma|}{m^{(t)}}+O(m^{(t)}\cdot \log \log n)
    \end{align*}
    bits. Since $\log\binom{|\Sigma|}{m^{(t)}}=O(m^{(t)}\log(|\Sigma|/m^{(t)}))$ and $|\Sigma|=n\poly\log n$, the space usage is at most $O(n/\log^cn)$ when $m^{(t)}\le n/\log^{10c}n$, and at most $O(m^{(t)}\log \log n)=O(H^{(t)}(\log \log n)/\log n+\log\log n)$ when $m^{(t)}>n/\log ^{10c}n$.
    \item The level-$2$ dictionary stores all symbols $\sigma$ with $f_{\sigma}^{(t)}>\lceil\log n\rceil$. Each symbol is associated with an $O(\log n)$-bit value, which is equal to $f_\sigma^{(t)}$. This dictionary uses $O(\log n)$ bits for every symbol $\sigma$ with $f_{\sigma}^{(t)}>\lceil\log n\rceil$. Since $f_\sigma\cdot \log(n/f_\sigma)=\Omega(\log^2 n)$ whenever $\log n<f_{\sigma}^{(t)}\le n/2$, an argument similar to \cref{clm:distinct_symbols_cheap} shows that the level-$2$ dictionary uses $O(H^{(t)}/\log n+\log n)$ bits of space.
\end{itemize}

The updates are straightforward: When the frequency of $\sigma$ increases, we first increase the value of $\sigma$ in the level-$1$ dictionary (and insert $\sigma$ if it did not exist already), then check if $f_\sigma^{(t)}>\log n$. If so, then we also update $\sigma$ in the level-$2$ dictionary. Decrements are handled similarly.

\subsubsection{Reducing the Size of the Lookup Tables}

In the algorithm of \cref{thm:arithmetic-code-alg}, the lookup tables of size $O(|\Sigma|\cdot \log^2n)$ are used to decode the singleton intervals of the encoding in \cref{thm:static-arithmetic-code-alg}. They essentially describe the encoding of symbols under a Shannon code \cref{lem:encode}. There are other lookup tables in the data structure (e.g., in the B-tree of \cref{lem:B-tree} and for decoding non-singleton intervals), but they have size $n^{\epsilon}$ and are not a concern. In order to reduce the size of the lookup table, we observe that for the very infrequent symbols where $f_\sigma=\poly \log n$, we are allowed to store each occurrence of $\sigma$ explicitly, using $O(\log|\Sigma|)$ bits each instead of $\log(n/f_\sigma^{(t)})$ bits (since $|\Sigma|=n\poly\log n$, these two values only differ by $O(\log \log n)$). Thus, it is possible to only encode the frequent symbols using Shannon codes, which means that the lookup table only has to decode the $n/\poly \log n$ frequent symbols, instead of all $|\Sigma|\cdot \log n$ symbols in a placeholder distribution. This way, the lookup tables cost $n/\poly\log n+n^\epsilon$ bits of space instead of $O(|\Sigma|\cdot \log^2n)+n^\epsilon$.

To achieve this, we need to make white-box modifications to the encoding method of \cref{thm:static-arithmetic-code-alg}, so that when handling a singleton interval, the algorithm first checks whether this interval encodes an infrequent symbol. If so, then we decode this symbol by directly reading its value (i.e., not using the lookup table). 

Note that in reality, it is infeasible to only store the infrequent symbols explicitly, since a symbol may switch from being infrequent to being frequent, at which time we cannot afford to re-encode all occurrences of this symbol. Instead, what we will do is to explicitly store the first (by order of insertion) $\poly\log n$ occurrences of \emph{every} symbol.

In the remainder of this section, we show how to make the necessary modifications to \cref{thm:static-arithmetic-code-alg}, and bound the total encoding length of our data structure.

\paragraph{Adding explicit symbols to the placeholder distribution.} Let $T_{\text{freq}}=\log^{10c}n$ be a threshold. In the new algorithm, we always encode the symbols under $\hat{D}^{(t)}$ for some $t$ (instead of $\widetilde{D}^{(t)}$), which is a modified placeholder distribution containing the following symbols:
\begin{itemize}
    \item \textbf{The original symbols:} For each $\sigma\in \Sigma$, if $f_\sigma^{(t)}\ge T_{\text{freq}}$ (i.e., it is relatively frequent), then $\hat{D}^{(t)}(\sigma)=f_\sigma^{(t)}/n$. Otherwise, $\hat{D}^{(t)}(\sigma)=0$.
    \item \textbf{The placeholder symbols:} For each $0\le k\le \lfloor \log n\rfloor$ \emph{such that} $2^{-(k-1)}\ge T_{\text{freq}}/n$ (note that we have fewer levels of placeholders than in $\widetilde{D}^{(t)}$), we add $\min(|\Sigma|,2^{k+1})$ level-$k$ placeholders $\tau_{k,1},\tau_{k,2},\dots$ to $\hat{D}^{(t)}$ just as in the original placeholder distribution $\widetilde{D}^{(t)}$.
    \item \textbf{The explicit symbols:} In addition, for every $\sigma\in \Sigma$, we add a symbol $\sigma_{\text{exp}}$ (the \defn{explicit symbol} of $\sigma$) to $\hat{D}^{(t)}$, which has probability $1/(|\Sigma|\cdot \log^{10c}n)$. In the encoding scheme of \cref{thm:static-arithmetic-code-alg}, our Shannon code only encodes the previous two types of symbols; the explicit symbols will be stored explicitly.
\end{itemize}
We normalize the probabilities to make $\hat{D}^{(t)}$ a proper distribution. Normalization only multiplies each probability by $1\pm O(1/\log^{10c}n)$, which is negligible.

\paragraph{Encoding the symbols under $\hat{D}^{(t)}$.} Note that, apart from the explicit symbols, the number of symbols in the support of $\hat{D}^{(t)}$ is at most $O(n/\log^{9c}n)$ (there are $O(\log n)$ levels of placeholders, and each level contains at most $O(n/T_{\text{freq}})=O(n/\log^{10c}n)$ symbols). Using this property, we show that it is possible to store an array of symbols under $\hat{D}^{(0)}$ using $H+O(H(\log \log n)/\log n+n/\log^cn)$ bits of space (that is, we can achieve the guarantees of \cref{thm:static-arithmetic-code-alg} but without the $O(|\Sigma|\cdot w)$ term in the space usage).

Recall how the symbols in a static distribution are encoded in \cref{thm:static-arithmetic-code-alg}: We partition each block into intervals and use a prefix code to encode each interval. For intervals that contain more than one symbol, lookup tables of size $2^{O(\epsilon\cdot \log n)}$ suffice. It remains to build a prefix code for encoding a symbol under $\hat{D}^{(t)}$, so that we can decode the prefix code using a lookup table of size $O(n/\log^cn)$. Moreover, we are allowed to waste $O(1)$ bits of space when encoding a symbol (see \cref{thm:static-arithmetic-code-alg}).

Let the first bit of the code indicate whether the symbol is explicit. If it is, then we use the next $\log|\Sigma|+O(1)$ bits to explicitly specify the symbol. Otherwise, we invoke \cref{lem:encode} on the original and placeholder symbols to construct a prefix code. Using the techniques in \cref{thm:static-arithmetic-code-alg}, this prefix code can be decoded using a lookup table of size $O(n/\log^{8c}n)$.

\paragraph{Representing the array using $\hat{D}^{(t)}$ in a space-efficient way.} 

Finally, we show how to represent the original array using symbols that are in the support of $\hat{D}^{(t)}$. For simplicity, we \emph{only describe what happens before the first rebuild}, where all symbols are encoded under $\hat{D}^{(0)}$. Let $\hat{a}^{(t)}_i$ denote the symbol in $\hat{D}^{(0)}$ that is used to represent $a_i^{(t)}$ (similar to $\widetilde{a}_i^{(t)}$): When $a_i^{(t)}=\sigma$, our algorithm may let $\hat{a}^{(t)}_i$ be the original symbol $\sigma$, some placeholder assigned to $\sigma$, or the explicit symbol $\sigma_{\text{exp}}$.

We maintain the invariant that for any symbol $\sigma$, the number of occurrences of the explicit symbol $\sigma_{\text{exp}}$ is at most $T_{\text{freq}}$. Later we will show that this invariant guarantees space efficiency. In order to maintain the invariant, we keep track of the number of explicit symbols, using a dictionary similar to the level-$1$ dictionary used to store the frequencies.

Initially, we let $\hat{a}_i^{(0)}=a_i^{(0)}$ if $f_{a_i^{(0)}}^{(0)}\ge T_{\text{freq}}$, and let $\hat{a}_i^{(0)}=(a_i^{(0)})_{\text{exp}}$ otherwise. After each update (which sets $a_i^{(t)}$ to $\sigma$), we first check if $f_\sigma^{(t)}> T_{\text{freq}}$. If so, we use a placeholder to encode $a_i^{(t)}$; otherwise, we let $\hat{a}_i^{(t)}=\sigma_{\text{exp}}$.

Although we have removed some levels of placeholders compared with the previous algorithm, it can be shown that those levels will not be used: When we need a placeholder corresponding to $\sigma$, it must be the case that $f_\sigma^{(t)}>T_{\text{freq}}$. Using the exact same criterion for choosing placeholders as in the previous algorithm, we ask for a level-$k$ placeholder where $2^{-(k-1)}\ge f_\sigma^{(t)}/n>T_{\text{freq}}/n$. Upon checking the definition of $\hat{D}^{(0)}$, we see that placeholders of this level exist.

It remains to show that the encoding is space-efficient.  Similarly to $\Delta^{(t)}$, we define $\hat{\Delta}^{(t)}$ as
\begin{align*}
    \sum_{i=1}^{n}\max\BK*{0, -\log \hat{D}^{(0)}(\hat{a}_i^{(t)})+\log D^{(t)}(a_i^{(t)})},
\end{align*}
where $-\log \hat{D}^{(0)}(\hat{a}_i^{(t)})$ is the cost of storing a symbol under $\hat{D}^{(0)}$ and $-\log D^{(t)}(a_i^{(t)})$ the cost of storing the same symbol under $D^{(t)}$.

\begin{claim}
    \label{clm:delta_decrement_slow_hat}
    After $t$ updates, we have $\hat{\Delta}^{(t)}=O(n/\log^{10c}n+t\log \log n+H^{(t)}(\log \log n)/\log n)$.
\end{claim}

Note that the bound of \cref{clm:delta_decrement_slow_hat} is slightly worse than \cref{clm:delta_decrement_slow}. However, since $O(H^{(t)}(\log \log n)/\log n)$ is always negligible, \cref{clm:delta_decrement_slow_hat} still implies space efficiency.

\begin{proof}
    Given \cref{clm:delta_decrement_slow}, we only need to show that $\hat{\Delta}^{(t)}-{\Delta}^{(t)}$ is small.

    Apart from the loss in normalization (which contributes at most $O(n/\log^{10c}n)$ to the difference), the difference between $\hat{\Delta}^{(t)}$ and ${\Delta}^{(t)}$ comes from the fact that some $\hat{a}_i^{(t)}$ may be explicit symbols. Indeed, whenever $\hat{a}_i^{(t)}$ is an original or placeholder symbol, $\widetilde{a}_i^{(t)}$ should be equal to $\hat{a}_i^{(t)}$, and we have
    \begin{align*}
        \max\BK*{0, -\log \hat{D}^{(0)}(\hat{a}_i^{(t)})+\log D^{(t)}(a_i^{(t)})}=\max\BK*{0, -\log \widetilde{D}^{(0)}(\widetilde{a}_i^{(t)})+\log D^{(t)}(a_i^{(t)})}\pm O(1/\log^{10c}n)
    \end{align*}
    where the last term is due to normalization. It remains to show that the difference due to explicit symbols is small. That is, we need to bound
    \begin{align*}
        &\sum_{1\le i\le n, \hat{a}_i^{(t)}\text{ is explicit}}\max\BK*{0, -\log \hat{D}^{(0)}(\hat{a}_i^{(t)})+\log D^{(t)}(a_i^{(t)})} \\
        \le{}&\sum_{f_\sigma^{(t)}\ne 0}\min\BK*{f_\sigma^{(t)},T_{\text{freq}}}\cdot \max\BK*{0,\log |\Sigma|+O(\log \log n)-\log (n/f_\sigma^{(t)})}.\tag{due to the invariant}
    \end{align*}
    Partition the symbols into two parts: The frequent ones with $f_\sigma^{(t)}\ge T_{\text{freq}}\cdot \log^{11c}n$, and the infrequent ones. The number of frequent symbols is at most $n/(T_{\text{freq}}\cdot \log^{11c}n)$, so they contribute at most
    \begin{align*}
        n/(T_{\text{freq}}\cdot \log^{11c}n)\cdot T_{\text{freq}}\cdot O(\log n)=O(n/\log^{10c}n)
    \end{align*}
    to the summation. As for the infrequent symbols, since $f_\sigma^{(t)}$ is small, we have that
    \begin{align*}
        \log |\Sigma|+O(\log \log n)-\log (n/f_\sigma^{(t)})=O(\log \log n).
    \end{align*}
    That is, we waste at most $f_\sigma^{(t)}\cdot O(\log \log n)$ for an infrequent symbol. In comparison, such $\sigma$ contributes at least $f_\sigma^{(t)}\cdot \Omega(\log n)$ to the entropy $H^{(t)}$. Therefore, the total contribution of infrequent symbols is at most $O(H^{(t)}\log \log n/\log n)$. This explains the last term in the statement of the claim.
\end{proof}

\subsection{Step 2: Alphabets of Size \texorpdfstring{$2^{O(w)}$}{2\^O(w)}}

Having shown how to support alphabets of size $n \polylog n$, we can now generalize to even larger alphabets, obtaining the main result of the paper:

\Main*

Note that when $|\Sigma|=n\poly\log n$, \cref{thm:arithmetic-code-poly-n} reduces to \cref{thm:arithmetic-code-n-poly-log-n}. This is due to \cref{clm:distinct_symbols_cheap}.

We obtain \cref{thm:arithmetic-code-poly-n} by adding a layer of indirection to \cref{thm:arithmetic-code-n-poly-log-n}, using a \defn{hashcode data structure} that maps symbols that currently exist in the array to distinct integers (which we call \defn{hashcodes}) in $[O(n \poly\log n)]$. This mapping is stable in the sense that the hashcode for a symbol $\sigma$ is fixed for the duration that $\sigma$ exists in the array (however, if $\sigma$ is completely removed from the array and later re-inserted, its new hashcode may be different from the old one). We then apply the data structure in \cref{thm:arithmetic-code-n-poly-log-n} to store the array of hashcodes, and use the hash table to translate between actual symbols and hashcodes. 

Our hashcode data structure is a simple variant of the dynamic perfect hashing construction in \cite{demaine2006dictionariis}. Note, however, that our data structure needs to map hashcodes \emph{back} to the original symbols, which was not required in the previous construction. Formally, the data structure achieves the following guarantees:

\begin{lemma}
    \label{lem:backward-stable-hashing}
    Let $n,U$ be integers. There exists a data structure in the word RAM model with word size $w=\Omega(\log n+\log U)$ that stores a set $S$ of keys from $[U]$ (it is guaranteed that $|S|\le n$ always holds), supporting the following operations in constant time:
    \begin{itemize}
        \item $\textup{\textsc{Insert}}(x)$: Insert a key into $S$.
        \item $\textup{\textsc{Delete}}(x)$: Delete a key from $S$.
        \item $\textup{\textsc{Forward}}(x)$: Given a key $x$, return a hashcode in $[O(n\poly\log n)]$. If $x$ is not in $S$, return $\bot$.
        
        It is guaranteed that at a fixed point in time, different keys in $S$ have different hashcodes. Also, the hashcode for a key $x$ is fixed for the duration that $x$ is in $S$ (however, the hashcode may change if $x$ is deleted from $S$ and then re-inserted).
        \item $\textup{\textsc{Backward}}(y)$: Given a hashcode $y$, return the corresponding key in $S$. If there is no such key, return $\bot$.
    \end{itemize}
    At any point in time, the data structure uses
    \begin{align*}
        \log\binom{U}{|S|}+O_c(|S|\cdot \log\log n)+O(n/(\log n)^c)
    \end{align*}
    bits of space for any constant $c$. Both the time and space bounds are with high probability in $n$.
\end{lemma}

Note that combining \cref{lem:backward-stable-hashing} with \cref{thm:arithmetic-code-n-poly-log-n} results in a data structure that supports $O(1)$-time operations and uses space
\begin{align*}
    H+\log\binom{|\Sigma|}{m}+O_c(m\log\log n+H(\log \log n)/\log n+n/(\log n)^c)
\end{align*}
bits. Using \cref{clm:distinct_symbols_cheap}, we know that $m\log\log n$ is at most $O(H(\log\log n)/\log n+\log\log n)$, so the space usage matches the bound desired by \cref{thm:arithmetic-code-poly-n}. It therefore remains only to prove \cref{lem:backward-stable-hashing}.
\begin{proof}[Proof of \cref{lem:backward-stable-hashing}]
    
    We obtain our data structure by slightly modifying the construction of \cite[Theorem 3]{demaine2006dictionariis}. 
    
    \paragraph{Hashing keys into buckets.} Let $b=n\cdot \log^{2c}n$ be a parameter, where $c$ is the constant in the lemma statement. We use $h$ to denote a quotient hash function drawn from the family of \cite[Theorem 5]{demaine2006dictionariis}, which is a bijective function mapping $[U]$ to $[b]\cdot [U/b]$ (we assume that $U$ is divisible by $b$; if not, we can increase $U$). Intuitively, $h$ sends the keys into $b$ buckets. We use $h(x)_1$ and $h(x)_2$ to denote the two components of $h(x)$.

    During the execution of our hash table, we say that a key $x$ is a \defn{frontyard key} if, \emph{at the time it is inserted}, no key $y\in S$ with $h(x)_1=h(y)_1$ is a frontyard key. Otherwise, we say that $x$ is a \defn{backyard key}. Note that if a key is deleted and then re-inserted, it may not be the same type of key anymore.
    
    It is shown in \cite{demaine2006dictionariis} that $h$ has the following properties:
    \begin{itemize}
        \item There exists an absolute constant $\alpha<1$, such that $h$ can be stored using $O(n^\alpha)$ bits of space, and that $h$ and $h^{-1}$ can be evaluated in constant time.
        \item Fix any point in time. With high probability in $n$, the current number of backyard keys in $S$ is at most $b'=O(n^2/b)=O(n/\log^{2c}n)$\footnote{This property was not formally stated in Theorem 5, but was established in Section 2.3 of \cite{demaine2006dictionariis}.}.
    \end{itemize}

    \paragraph{Assigning hashcodes.} We assign the hashcodes as follows: For each frontyard key, its hashcode is simply the index of its bucket, which is an integer in $[b]$. As for the backyard keys, we assign an integer in $[b+1,b+b']$ to each of them. To achieve this, we use a queue to keep track of the unassigned hashcodes. Initially, every integer in $[b+1,b+b']$ is in the queue. When a backyard key is inserted, we take an integer from the queue as its hashcode. Similarly, when a backyard key is removed, we insert its hashcode back into the queue. This queue only uses $o(n/\log^cn)$ bits of space.

    \paragraph{Storing the mapping.} We store the frontyard and backyard keys separately using the resizable dictionaries of \cite{bender2022optimal}.
    \begin{itemize}
        \item In the frontyard dictionary, the universe is of size $b$ (i.e., the set of buckets), and each key is associated with a value in $[U/b]$. For each frontyard key $x$, we insert $h(x)_1$ in the dictionary with $h(x)_2$ being its associated value. Let the current number of frontyard keys be $m'$, this dictionary uses
        \begin{align*}
            \log\binom{b}{m'}+m'\cdot \log(U/b)+O(m'\log\log n)
        \end{align*}
        bits, which is at most
        \begin{align*}
            \log\binom{U}{m'}+O(m'\log\log n).
        \end{align*}
        \item For the backyard keys: We use two backyard dictionaries, one mapping keys to hashcodes and one mapping hashcodes back to keys, wasting $O(\log n)$ bits per backyard key. This uses $o(n/(\log n)^c)$ bits.
    \end{itemize}

    We concatenate the dictionaries and the queue together using \cref{lem:polylog-concat}. In total, the data structure uses
    \begin{align*}
        \log\binom{U}{|S|}+O_c(|S|\cdot \log\log n)+O(n/(\log n)^c)
    \end{align*}
    bits of space.

    \paragraph{Answering queries.} The queries are answered as follows:
    \begin{itemize}
        \item For a forward query on key $x$, we first check the frontyard dictionary. If $h(x)_1$ exists in the frontyard dictionary and has value $h(x)_2$, then we return $h(x)_1$. If not, we check whether the backyard dictionaries contain $x$, and act accordingly.
        \item For a backward query, if $y$ is a frontyard hashcode (i.e., $y\in [b]$), then we query the frontyard dictionary to see if $y$ exists. If it exists and has value $v$, then we return $h^{-1}(y,v)$, otherwise we return $\bot$. The backyard hashcodes are handled by directly querying the backyard dictionaries.\qedhere
    \end{itemize}
\end{proof}

%% file: lower-bound.tex
\section{Lower Bounds}\label{sec:lower}

In this section, we present two lower bounds, showing that, in the parameter regimes specified below, the multiplicative redundancies of our data structures (both for entropy-encoded arrays and for prefix-coded arrays) are tight up to a $\log\log n$ factor.

We start with a lower bound for entropy-encoded arrays, showing that in the parameter regime of $|\Sigma| = O(\sqrt{n})$, and even if we condition on $H$ having a prescribed value up to a constant factor in the range $[n/\log^{O(1)}n, (1/100)n \log n]$, any constant-time solution must use space $(1 + \omega(1 / \log n)) \cdot H$.

\LbMain*

Our proof is by reducing from key-value dictionaries. Given a data structure for storing an entropy-encoded array of length $n$, we use it to implement a dynamic dictionary of universe size $n$. In the array, we let the $i$-th entry be $\bot$ if key $i$ is not in the dictionary, or let it be the value associated with key $i$ otherwise. We can then apply lower bounds for dynamic dictionaries to obtain lower bounds for entropy-encoded arrays. 

It is shown in \cite{li2023tight} that the following distribution of inputs is hard for key-value dictionaries:

\newcommand{\keyset}{K}
\newcommand{\insa}{a} %
\newcommand{\insai}[1][i]{a_{#1}} %
\newcommand{\deli}[1][i]{d_{#1}} %

\begin{algorithm}[H]
  \captionof{Distribution}{Hard Distribution}
  \label{alg_hard_dist}
  \DontPrintSemicolon
  \SetKwFunction{Query}{Query}
  \SetKwFunction{Insert}{Insert}
  \SetKwFunction{Delete}{Delete}
  Initialize an empty dictionary with capacity $n$, key-universe $[U]$ and value range $[V]$\;
  $\keyset \gets$ uniform random $n$-element subset of $[U]$\;
  Insert keys in $\keyset$ into the dictionary one by one, using $n$ insertions. The value associated to each key is randomly sampled from $[V]$\;
  \For{$i = 1$ to $n$} {
    $\deli[i] \gets$ a uniform random key in $\keyset$ that has not been removed\label{step_loop_begin}\;
    \Query{$\deli[i]$}\;
    \Delete{$\deli[i]$} from the dictionary\;
    $\insai[i] \gets$ a uniform random key in $[U]$ which is neither in $\keyset$ nor in the current dictionary\;
    \Insert{$\insai[i]$} to the dictionary\label{step_loop_end}. The value associated to $a_i$ is randomly sampled from $[V]$\;
  }
\end{algorithm}

\begin{theorem}[{\cite[Theorem 1.3]{li2023tight}}]
    \label{thm:li2023tight}
    Let $U\ge 3n$ and $V=U^{2+\Theta(1)}/n^2$. Suppose that there exists a dynamic dictionary in the word RAM model with word size $w=\Theta(\log U)$ that, for any sequence of operations sampled from \cref{alg_hard_dist}, performs each operation in constant time and correctly answers the queries (with high probability over the internal randomness of the dictionary). Then such a data structure must use
    \begin{align*}
        \log\binom{U}{n}+n\log V+\Omega(n \log^*n)
    \end{align*}
    bits of space in expectation.
\end{theorem}
We remark that $\log^*n$ can be replaced by any function that is smaller than $\log^{(k)}n$ for any constant $k$, where $\log^{(k)}n$ is the $k$-th iterated logarithm.

\begin{proof}[Proof of \cref{thm:arith_lb}]
    Assume for the sake of contradiction that there exists a data structure for storing an entropy-encoded array that uses $H+O(H\sqrt{\log ^*n}/\log n)$ bits of space when $H\in [0.9H_0,1.1H_0]$. We show that such a data structure can be used to implement a key-value dictionary. Let $U'=n$ be the universe size of the dictionary, and let $V'=n^{0.5}$ be the value range. We choose the number of keys $n'$ carefully, so that
    \begin{align*}
        \log\binom{U'}{n'}+n'\log V'=H_0\pm O(\log n).
    \end{align*}
    This implies that we should have $n'\le H_0/\log V'=2H_0/\log n$. Since $H_0\le (1/100)n\log n$, we always have that $n'\le n/3=U'/3$, which means that our parametrization satisfies \cref{thm:li2023tight}.

    Let $\Sigma=[|\Sigma|]$ denote an alphabet, where $|\Sigma|=V'+1$. To implement the dictionary, we store the following array in an entropy-encoded fashion, where each entry is in $\Sigma$: We identify each entry in our array with a key in the universe $[U']=[n]$. For the $i$-th entry, if $i$ exists in the dictionary, then $a_i$ is equal to the value of key $i$; otherwise $a_i=V'+1$. 
    
    Under this reduction, we can see that operations on the dictionary are easily translated into operations on the array, so it remains to bound the space usage. Fix a point in time, and let $s=\sum_{\sigma\le V'} f_\sigma$ denote the current number of keys in the set. The current entropy $H$ of the array can be bounded as
    \begin{align*}
        H&=\sum_\sigma f_\sigma\log(n/f_\sigma) \\
        &=f_{V'+1}\cdot \log(n/(n-s))+\sum_{\sigma\le V'}f_\sigma\cdot(\log(n/s)+\log(s/f_\sigma))\\
        &=\bk*{(n-s)\cdot\log(n/(n-s))+s\cdot\log(n/s)}+\sum_{\sigma\le V'}f_\sigma\log(s/f_\sigma).
    \end{align*}
    For inputs drawn from \cref{alg_hard_dist}, we always have $s=n'-1$ or $n'$. Therefore, the first term is always in the range (note that $\log(1+x)<x$ when $x>0$, and that $U'=n$)
    \begin{align*}
        n'\cdot \log(U'/n')+[-O(\log n),n'].
    \end{align*}
    As for the latter summation, a direct Chernoff bound shows that when the input is drawn from \cref{alg_hard_dist}, $f_\sigma$ is close to $s/V'$ for any $1\le \sigma\le V'$. Since $\log\binom{U'}{n'}+n'\log V'=H_0\pm O(\log n)$ by definition, we thus have that, with high probability over a random sequence drawn from \cref{alg_hard_dist}, $H\in H_0\pm O(n')=H_0\cdot (1\pm O(1/\log n))$ holds (at any point during the execution), which also implies that $H\in [0.9H_0,1.1H_0]$. By our assumption on the space-efficiency of our data structure, we obtain a dynamic dictionary that, on inputs drawn from \cref{alg_hard_dist}, uses $H_0+O(H_0\sqrt{\log^*n}/\log n)$ bits of space in expectation.

    In comparison, the lower bound of \cref{thm:li2023tight} implies that the data structure must use at least
    \begin{align*}
        &\log\binom{U'}{n'}+n'\log V'+\Omega(n'\log^*n')=H_0+\Omega(H_0\log^*n/\log n)
    \end{align*}
    bits of space, which is a contradiction.
\end{proof}

We also prove a similar lower bound in the prefix-code setting:

\begin{theorem}
    \label{thm:prefix_lb}
    Let $n$ be an integer. Let $H_0\in [100\cdot n,(1/100)\cdot n\log n]$ be a parameter of our choice. There exists a prefix code $\Sigma$ containing codewords of length at most $O(\log n)$ such that: Suppose that there exists a data structure in the word RAM model with $w=\Theta(\log n)$ that stores an array $a_1,\dots,a_n$ of $n$ symbols from $\Sigma$, and supports updates and queries in worst-case constant time. Then such a data structure must use
    \begin{align*}
        H+\omega(H/\log n)
    \end{align*}
    bits of space, where $H=\sum_{i=1}^{n}|a_i|$ is the sum of the lengths of the codewords. The lower bound holds even when we require that $H\in [(1/100)\cdot H_0,100\cdot H_0]$ at all times.
\end{theorem}

If we do not require that $H\in [(1/100)\cdot H_0,100\cdot H_0]$, then the exact same proof strategy as \cref{thm:arith_lb} works: We reduce from the key-value dictionary where $U'=n$, $n'=n/2$ and $V'=n^{0.5}$. For each $1\le i\le n$, the $i$-th codeword in our array is $0$ if key $i$ does not exist in the dictionary, and is $1\circ \text{str}(v)$ if $i$ exists and is associated with value $v\in [V]$ ($\text{str}(v)$ denotes the binary representation of $v$).

This strategy runs into trouble when we consider the case where $H\in [(1/100)\cdot H_0,100\cdot H_0]$, for instance when it is required that $H=O(n)$. If we wish to apply the same reduction as in \cref{thm:arith_lb}, then we must reduce from a dictionary that only contains $O(n/\log n)$ keys. However, in this reduction, the keys that are not in the dictionary contribute at least $\Omega(n)$ to the total space usage (because each entry must have length at least $1$), which is too costly. To handle this problem, we consider a generalization of \cref{thm:li2023tight}, which considers ``augmented dictionaries'', i.e., dictionaries that are stored together with an array.

\begin{corollary}
    \label{cor:li2023tight_augmented}
    Suppose that there exists a data structure that not only stores a dictionary as stated in \cref{thm:li2023tight}, but also stores a random variable $S$ that is independent of the dictionary and immutable (that is, $S$ is sampled during preprocessing and is fixed throughout the execution). The data structure only has to be able to recover $S$ at any point in time; that is, we do not require efficient access. Then such a data structure must use
    \begin{align*}
        \log\binom{U}{n}+n\log V+H(S)+\Omega(n\log^*n)
    \end{align*}
    bits of space in expectation. Here $H(S)$ denotes the binary entropy of $S$.
\end{corollary}
In other words, the dictionary does not gain any advantage when stored together with $S$.

\begin{proof}
    We only describe the necessary modifications to the proof of \cref{thm:li2023tight}. Firstly, we modify the hard distribution \cref{alg_hard_dist}, by additionally sampling $S$ at the beginning and storing it in the data structure.
    
    Secondly, we modify the definitions of the communication games. \cref{thm:li2023tight} was proven using two one-way communication games: The \emph{outer game} and the \emph{inner game}. It was shown that, if there exists an impossibly good dictionary, then this dictionary can be used to design a protocol for at least one of the communication games, that sends fewer bits than the information-theoretic optimum of the game, from which we can derive a contradiction.

    The definition of the outer game needs no modification. As for the inner game, we additionally require that Bob also recovers $S$ at the end. This is easy for Bob, since he recovers the data structure state $C_{\text{bef}}$ at the end of the protocol, and can directly obtain $S$. Moreover, this additional requirement retains the guarantee that the total length of communication is small compared to the information-theoretic optimum needed to send the input, which gives us the desired contradiction and proves the lower bound.
\end{proof}

\begin{proof}[Proof of \cref{thm:prefix_lb} using \cref{cor:li2023tight_augmented}]
    Similar to the proof of \cref{thm:arith_lb}, assume for the sake of contradiction that there exists a data structure for storing a prefix-coded array that uses $H+O(H\sqrt{\log ^*n}/\log n)$ bits of space when $H\in [(1/100)H_0,100\cdot H_0]$. Consider the following reduction from the augmented dictionary of \cref{cor:li2023tight_augmented} to storing an array of prefix codes.

    Let $k_{\text{max}}=O(\log\log n)$ be an integer to be determined later. Let $X_i$ ($1\le i\le n$) be i.i.d. random variables sampled from $[k_{\text{max}}]$, where $\Pr[X_i=k]=2^{-k}/(1-2^{-k_{\text{max}}})$. Note that
    \begin{align*}
        H(X_i)&=\sum_{k=1}^{k_{\text{max}}}-\Pr[X_i=k]\cdot \log \Pr[X_i=k] \\
        &=\bk*{1+ O\bk*{2^{-k_{\text{max}}}}}\cdot \bk*{\sum_{k=1}^{k_{\text{max}}}2^{-k}\cdot \bk*{k-O\bk*{2^{-k_{\text{max}}}}}} \\
        &=\Theta(1).
    \end{align*}
    Therefore the overall entropy $\sum_{i=1}^{n}H(X_i)=\Theta(n)$. 
    
    We initialize \cref{cor:li2023tight_augmented} with the following parameters: Let $U'=n$, $n'= U'/2^{k_{\text{max}}}$ and $V'=n^{0.5}$, where we assume that $n$ and $V'$ are powers of $2$ (which implies that $n'$ is an integer). We sample $X_1,\dots,X_n$, and let $S=(X_1,\dots,X_n)$ be the additional random variable as defined in \cref{cor:li2023tight_augmented}. The information-theoretic optimum of this augmented dictionary is
    \begin{align*}
        \text{OPT}=\log\binom{U'}{n'}+n'\log V'+\sum_{i=1}^{n}H(X_i).
    \end{align*}
    We set $k_{\text{max}}$ carefully so that $\text{OPT}\in [(1/4)H_0,4H_0]$. Since $H_0\in [100\cdot n,(1/100)\cdot n\log n]$, it can be shown that this is always possible, and that $k_{\text{max}}=\log\Theta(n\log n/H_0)=O(\log\log n)$.

    We store the information as follows: For $1\le i\le n$, if $i$ does not exist in the dictionary, then the $i$-th codeword is $\underbrace{000\cdots0}_{X_i - 1}1$. Otherwise, if the $i$-th key exists and has value $v\in [V]$, then the $i$-th codeword is $\underbrace{000\cdots0}_{k_{\text{max}}}\circ \,\text{str}(v)\circ\underbrace{000\cdots0}_{X_i - 1}1$, where $\text{str}(v)$ denotes the binary representation of $v$. When inserting or deleting a key, we simply have to update the corresponding entry in the array (note that the $X_i$'s are immutable). It is not hard to see that this is a valid implementation of an augmented dictionary.

    We bound the expected total number of bits as follows:
    \begin{align*}
        \E[H]=\E\Bk*{n'\cdot k_{\text{max}}+n'\cdot\log V'+\sum_{i=1}^{n}X_i}.
    \end{align*}
    Fixing one $X_i$, we have that
    \begin{align*}
        \E[X_i]=\sum_{k=1}^{k_{\text{max}}}\Pr[X_i=k]\cdot k=H(X_i)+O\bk*{2^{-k_{\text{max}}}}.
    \end{align*}
    Therefore the expected total size is (where we use the fact that $k_{\text{max}}=\log \Theta(n\log n/H_0)$)
    \begin{align*}
        \E[H]=n'\cdot k_{\text{max}}+n'\cdot\log V'+\sum_{i=1}^{n}H(X_i)+O(H_0/\log n).
    \end{align*}
    Further note that
    \begin{align*}
        \log\binom{U'}{n'}=n'\cdot \log\bk*{U'/n'}+O(n'),
    \end{align*}
    which is because of the standard inequality $(U'/n')^{n'}\le \binom{U'}{n'}\le (eU'/n')^{n'}$. Plugging in the definition $n'= U'/2^{k_{\text{max}}}$, we have that $\log\binom{U'}{n'}=n'\cdot k_{\text{max}}+O(n')$. We therefore have that $\E[H]\in [(1/4)H_0-O(H_0/\log n),4H_0+O(H_0/\log n)]$. Similar to the proof of \cref{thm:arith_lb}, we can apply a concentration bound on $\sum_{i=1}^{n}X_i$ to show that it is concentrated around $\sum_{i=1}^{n}\E[X_i]=\sum_{i=1}^{n}H(X_i)+O(n2^{-k_{\text{max}}})$. Since $n2^{-k_{\text{max}}}=O(H_0/\log n)$, we have that $H\in [(1/100)H_0,100 H_0]$ with high probability. Therefore, under our assumption, our data structure only uses
    \begin{align*}
        n'\cdot k_{\text{max}}+n'\cdot\log V'+\sum_{i=1}^{n}H(X_i)+O(H_0\sqrt{\log^*{n}}/\log n)
    \end{align*}
    bits of space in expectation. On the other hand, the lower bound of \cref{cor:li2023tight_augmented} is
    \begin{align*}
        n'\cdot k_{\text{max}}+n'\cdot\log V'+\sum_{i=1}^{n}H(X_i)+\Omega(n'\log^*{n'})
    \end{align*}
    bits, where the last term can be written as $\Omega(H_0\log^*n/\log n)$. This is a contradiction.
\end{proof}

%% file: applications.tex
\section{Applications}\label{sec:applications}

Finally, in this section, we present several applications of entropy-encoded arrays.

\subsection{Succinct dynamically resizable dictionary} 

Using entropy-encoded arrays as a building block, we now describe how to construct a dynamic dictionary, storing $n$ keys from a universe $[U]$, that offers a succinct space guarantee as a function of the current value of $n$, so long as $1\le n\le U/2$.

Note that, when $n \ll U$, one can use a previous construction by \cite{bender2022optimal} that uses $\log\binom{U}{n}+O(\log^{(k)}n)$ bits of space for $k=O(1)$ and is resizable. This construction fails to be succinct only when $n$ gets relatively close to $U$. Although there do exist succinct dynamic dictionaries in the regime where $n$ is close to $U$ (see e.g., \cite{bender2022optimal, arbitman2010backyard,raman2003succinct}), the known constructions either fail to be resizable or offer only amortized-expected time guarantees. Our construction is the first to offer high-probability $O(1)$-time guarantees while being resizable and succinct, even as $n$ gets asymptotically close to $U$.

\begin{theorem}
    \label{thm:succinct_resizable_dict}
    Let $U$ be an integer, and let $V \le \poly U$. There exists a data structure in the word RAM model with word size $w=\Omega(\log U)$ that maintains a key set $S\subseteq[U]$ of size $|S|\le U/2$ with each key corresponding to a value in $[V]$, while offering the following guarantees. The data structure supports insertions, deletions and membership queries in worst-case constant time. At any point, the data structure uses
    \begin{align*}
        (1+O(\log\log\log U/\log\log U))\cdot\bk*{\log\binom{U}{|S|}+|S|\cdot \log V}
    \end{align*}
    bits of space. At any point, both the space and time bounds are with high probability in $|S|$.
\end{theorem}

\begin{proof}
    We first present a dictionary that runs in amortized constant time, and then show how to deamortize it.
    
    \paragraph{Combining two cases.} Our dictionary is obtained by combining two techniques:
    \begin{itemize}
        \item \textbf{Small mode:} When $|S|=O(U/\log^2U)$, the keys are stored using the dictionary of \cite{bender2022optimal}, costing $\log\binom{U}{|S|}+|S|\cdot \log V+O(|S|\log\log \log U)$ bits of space. The extra $O(|S|\log\log \log U)$ term is acceptable since $(\log\log \log U/\log\log U)\cdot \log\binom{U}{|S|}=(\log\log \log U/\log\log U)\cdot \Omega(|S|\cdot \log(U/|S|))=\Omega(|S|\log\log \log U)$. 
        \item \textbf{Large mode:} When $|S|=\Omega(U/\log^2U)$, we use \cref{thm:arithmetic-code-poly-n} to store an array of length $U$, where the $i$-th entry is $v$ if $i\in S$ and $v$ is the value associated with $i$, and $V+1$ otherwise. This costs
        \begin{align*}
            (1+O(\log \log U/\log U))\bk*{\log\binom{U}{|S|}+|S|\cdot \log V}+U/\poly\log U
        \end{align*}bits of space (this is exactly the same as the reduction in the lower bound proof \cref{thm:arith_lb} when $V=\sqrt{U}$, and is a simple generalization when $V\gg U$). Here, the first term dominates the second because $|S|=\Omega(U/\log^2 U),|S|\le U/2$. Note that this space bound holds even when $|S|$ is small. The only problem is that the $U/\poly\log U$ term may be too large compared with $|S|$.
    \end{itemize}

    During the execution, as $|S|$ changes, we may need to switch between the two modes. Initially, the set $S$ is empty, and we are in the small mode. After handling an update, if we are currently in the small mode and $|S|>2U/\log^2U$, then we switch to the large mode. Similarly, if we are currently in the large mode, a switch to the small mode is performed when $|S|<U/\log^2U$.

    Note that by our definition, there must be $U/\log^2U$ operations between any two switches (indeed, after we switch to the large case, we need to delete at least $U/\log^2U$ keys to initiate the next switch, and vice versa). Thus, to show that switches take constant amortized time, it remains to show that each switch can be performed in $\Theta(U/\log^2 U)$ time.

    \paragraph{Performing switches in $\Theta(U/\log^2 U)$ time.} When we switch from the large case to the small case, we first enumerate all the keys currently stored in the array by traversing the B-tree (\cref{lem:B-tree}) of every chunk. For each chunk, after enumerating all the keys, we remove the keys from the array and move them to the dictionary. This costs $O((\text{\# of keys})+(\text{\# of chunks}))$ time. When we perform a switch, it must be the case that the number of keys is $|S|=\Theta(U/\log^2U)$, and the number of chunks can be set to be $o(U/\log^2U)$. Thus, the switch can be performed in $O(U/\log^2U)$ time.

    A switch from the small case to the large case is handled similarly. The only difference lies in how we enumerate the keys. Note that, for the dictionary in \cite{bender2022optimal}, one can straightforwardly implement an additional operation $\texttt{RandomChoice}$, which returns a uniformly random key in the dictionary in constant time. Using this, we can enumerate the keys by repeatedly applying $\texttt{RandomChoice}$, and removing the returned key from the dictionary.
    
    Note that it is essential for the above construction that $\texttt{RandomChoice}$ returns a random key, since otherwise, when we remove the returned keys from the dictionary, our choice of the deleted key might depend on the internal randomness of the dictionary. In this case, there would no longer be any guarantee on the performance of the dictionary.
    
    \paragraph{Deamortization.} To deamortize the switches, we use the same approach as in \cref{thm:arithmetic-code-alg}. That is, we run the switches in the background, using $O(1)$ time per operation. Since there must be $U/\log^2U$ operations between any two switches, we can guarantee that each switch is fully completed before the next one starts.

    As for the space usage: We show that, at any point during a switch, our data structure is space-efficient. At any point during a switch, each key is either stored in the small mode or in the large mode. When $k$ keys are stored in the small mode and $|S|-k$ are in the large mode, the total space usage is
    \begin{align*}
        &(1+O(\log\log \log U/\log\log U))\cdot \Big(
        \log\binom{U}{k}+k\cdot \log V \\
        &\qquad+\log\binom{U}{|S|-k}+(|S|-k)\cdot \log V\Big)+U/\poly\log U,
    \end{align*}
    where we can ignore the $U/\poly\log U$ term because we always have $|S|=\Theta(U/\log^2U)$ when undergoing a switch. It remains to compare $\log\binom{U}{k}+\log\binom{U}{|S|-k}$ with $\log\binom{U}{|S|}$. Using the classical inequality $(n/k)^k\le \binom{n}{k}\le (en/k)^k$, we have
    \begin{align*}
        \log\binom{U}{k}+\log\binom{U}{|S|-k}&\le O(|S|)+k\log(U/k)+(|S|-k)\log (U/(|S|-k)) \\
        &\le O(|S|)+|S|+|S|\log (U/|S|),\tag{Jensen's inequality}
    \end{align*}
    and
    \begin{align*}
        \log\binom{U}{|S|}\ge |S|\log (U/|S|).
    \end{align*}
    We thus have
    \begin{align*}
        \log\binom{U}{k}+\log\binom{U}{|S|-k}=(1+O(1/\log \log U))\log\binom{U}{|S|},
    \end{align*}
    which concludes the proof of space-efficiency.
    
\end{proof}

\subsection{Space-efficient fingerprint filters}

Our second application is to constructing succinct $O(1)$-time fingerprint filters in the parameter regime where $\epsilon^{-1} \le \polylog n$ and $1-\epsilon=\Omega(1)$. Recall that, in a fingerprint filter, every key in some set $S$ is hashed to a fingerprint in $[n\delta^{-1}+O(1)]$\footnote{The additive $O(1)$ term is not just for rounding $n\delta^{-1}$ to an integer. If we have exactly $n\delta^{-1}$ fingerprints, then our false positive rate would be slightly larger than $\epsilon$. This $O(1)$ is not large: For sufficiently large $n$, $\lceil n\delta^{-1}+1\rceil$ fingerprints suffice.}, where $\delta = \ln(1/(1 - \epsilon))$. This guarantees that, for a key $x\in [U]$ not in the set $S$, the probability that some key in $S$ shares its fingerprint is at most $\epsilon$. The fingerprint filter is then responsible for storing the multiset of fingerprints, while supporting insertions, deletions, and membership queries.

Let $C_i$ be the number of elements with fingerprint $i$. Then the distribution of the variables $C_i$ is very closely approximated by $\text{Pois}(\delta |S| / n)$. We can obtain a succinct filter by simply storing the variables $C_i$ directly in an entropy-encoded array. This is captured by the following theorem:

\begin{theorem}
    \label{thm:filter_ub}
    Let $n$ be an integer, let $U=\poly n$ and let $\epsilon=\Omega(1/\poly\log n)$, with $\epsilon \in (0, 1)$, be a parameter such that $1-\epsilon=\Omega(1)$. There exists a dynamic fingerprint filter in the word RAM model with word size $w=\Omega(\log U)$ that maintains a set $S\subseteq[U]$ of size at most $n$. This filter supports insertions, deletions and approximate membership queries, where the false positive rate is at most $\epsilon$, by storing a multiset of fingerprints in the range $[n \delta^{-1}+O(1)]$, where $\delta=\ln(1/(1-\epsilon))$. (Note that such a $\delta$ satisfies $\epsilon\le \delta=O(\epsilon)$.) The filter uses
    \begin{align}
        \label{eq:filter_bound}
        (1+O(\log\log n/\log n))\cdot n\delta^{-1}H(\text{Pois}(\delta))+n/\poly\log n
    \end{align}
    bits of space, and answers queries in worst-case constant time. Here $H(\text{Pois}(\delta))$ denotes the binary entropy of $\text{Pois}(\delta)$. Both the space and time bounds are with high probability in $n$.
\end{theorem}

To get some intuition for the space bound of $n\delta^{-1}H(\text{Pois}(\delta))$, we present the following claim. The proof is deferred to \cref{app:poisson}.

\begin{restatable}{claim}{Poisson}
    \label{clm:connect_delta_with_eps}
    For any $0<\epsilon<0.9$, we have
    \begin{align*}
        n\delta^{-1}H(\text{Pois}(\delta))\in [n\log\epsilon^{-1}, n\log\epsilon^{-1}+n\log e].
    \end{align*}
\end{restatable}

\begin{proof}[Proof of \cref{thm:filter_ub}]
    We first present a filter that assumes access to free randomness, and achieves the space bound \eqref{eq:filter_bound} only in expectation. Later, we will remove these assumptions using Dietzfelbinger's splitting trick \cite{Dietzfelbinger2009Applications, Siegel2004Universal}.
    
    Fix $n,\epsilon$, and let $m$ be the smallest integer such that $(1-1/m)^{n}\ge 1-\epsilon$. We use a random hash function $h:[U]\to[m]$ to hash the keys into $m$ hashes. Our data structure then stores an array of $m$ entries in an entropy-encoded fashion, where the $i$-th entry stores $C_i$, which is the number of real keys that are hashed to $i$.

    Insertions and deletions are handled by updating the value $C_{h(x)}$ (where $x$ is the queried key). When answering a membership query on a key $x\in [U]$, we check the value $C_{h(x)}$, and output yes if $C_{h(x)}\ne 0$. To see why we need $(1-1/m)^{n}\ge 1-\epsilon$, note that if $x$ is not in the set, then the probability of us falsely reporting a positive is
    \begin{align*}
        1-\bk*{1-\frac 1m}^n\le \epsilon,
    \end{align*}
    which is acceptable.

    We store the array $C_1,\dots,C_m$ using \cref{thm:arithmetic-code-poly-n}. Since \cref{thm:arithmetic-code-poly-n} encodes the array under the empirical distribution, we have that for any distribution $D$ over the symbols $\{0,1,\dots,n\}$, the encoding length $H$ of \cref{thm:arithmetic-code-poly-n} is at most $\sum_{i=1}^{m}-\log D(C_i)$. Since each $C_i$ is identically distributed, we can set $D$ to be the distribution of $C_1$, and bound $H$ as $m\cdot H(C_1)$.
    
    When $n$ is large, it can be shown that $m=n\delta^{-1}+O(1)$, so the distribution of each $C_i$ is closely approximated by $\text{Pois}(\delta)$. In particular, we have $H(C_i)=(1\pm O(1/\poly\log n))\cdot H(\text{Pois}(\delta))$. Thus, we can bound the cost of storing the array as
    \begin{align*}
        (1+O(\log\log n/\log n))\cdot n\delta^{-1}H(\text{Pois}(\delta))+n/\poly\log n
    \end{align*}
    bits in expectation. We have $C_i=O(\log n)$ for every $i$ with high probability; therefore, the cost of storing nonzero frequencies in the space bound of \cref{thm:arithmetic-code-poly-n} is negligible.

    Finally, we remove the two assumptions made at the start of the proof.
    \begin{itemize}
        \item To improve the space guarantee to high probability, we first hash the keys into $n/\log^{10}n$ buckets. With high probability in $n$, every bucket receives at most $n'=\log^{10}n+\log^{8}n$ keys. We then build a filter for each bucket (built to handle up to $n'$ keys). These filters are concatenated using \cref{lem:polylog-concat}. Since the filter of each bucket independently satisfies the bound of \eqref{eq:filter_bound} (with $n$ replaced by $n'$) in expectation, and uses space at most $\polylog n$ with high probability, a concentration bound over all buckets shows that the total space usage satisfies \eqref{eq:filter_bound} with high probability.
        \item Building on top of the previous improvement, to remove the requirement of fully random hash functions, we hash the keys into $n^{1/2}$ mega-buckets (using an $O(1)$-time $\poly\log n$-wise independent hash function from \cite{Siegel1989Universal}), and build a filter for each mega-bucket, which can handle $n''=n^{1/2}+n^{1/3}$ keys. Note that we use the improved filter for each mega-bucket, which assumes access to free randomness and satisfies \eqref{eq:filter_bound} (with $n$ replaced by $n''$) with high probability. Next, we sample an $O(n^{1/2})$-wise independent hash function, and use it in place of the random hash functions of the filters (i.e., we use the same hash function for each mega-bucket). This removes the requirement for random hash functions, and also preserves the high-probability space guarantee.
    \end{itemize}
    This concludes our construction.
\end{proof}

\subsection{Quotient filters} 

Finally, we also give an alternative filter construction, in which we show how to turn quotient filters---a classical way of building a fingerprint filter that is not, a priori, succinct---into \emph{succinct} fingerprint filters by simply storing the filter directly in an entropy-encoded array.

\begin{theorem}
     Let $n$ be an integer, let $U=\poly n$ and let $1/\sqrt{n}<\epsilon<1$ be a parameter such that $1-\epsilon=\Omega(1)$. There exists a modified quotient filter in the word RAM model with word size $w=\Omega(\log U)$ that maintains a set $S\subseteq[U]$ of size at most $n$, supporting insertions, deletions and approximate membership queries, where the false positive rate is at most $\epsilon$. The filter uses
    \begin{align}
        \label{eq:quotient_filter_bound}
        &\bk*{1+O\bk*{\frac 1{\sqrt{\log n}}}}\cdot n\delta^{-1}H(\text{Pois}(\delta))
    \end{align}
    bits of space, with high probability in $n$, and answers queries in expected constant time. Here $\delta=\ln(1/(1-\epsilon))$ as in \cref{thm:filter_ub}, and $H(\text{Pois}(\delta))$ denotes the binary entropy of $\text{Pois}(\delta)$.
\end{theorem}

To simplify our exposition in the proof below, we will assume access to free randomness, and we will settle for expected-space bounds rather than high-probability ones. These relaxations can then be removed using the same splitting trick as was used in Theorem \ref{thm:filter_ub}.

\begin{proof}

    We first describe the structure of a typical quotient filter, then show how to improve its space usage.

\paragraph{Background: how quotient filters work.} A quotient filter hashes the keys in $[U]$ to \defn{fingerprints} in $[n\delta^{-1} + O(1)]$ using a random hash function $h:[U] \to [n\delta^{-1}+O(1)]$, and stores the multiset of fingerprints. The filter further partitions each fingerprint into two parts: Write $h(x)$ (uniquely) as $h_1(x)+(h_2(x)-1)\cdot 2n$, such that $h_1(x)\in [2n]$ and $h_2(x)\in [2\delta^{-1}+O(1)]$. We call $h_1(x)$ the \defn{hash} of $x$, and $h_2(x)$ the \defn{remainder} of $x$.

A quotient filter uses an ordered linear probing hash table of some size $(1 + \beta)n$ slots to store the fingerprints---in a classical quotient filter, one tends to select $\beta$ to be small (so that the filter is space efficient), but for our purposes it will be acceptable to simply use $\beta= 1$, meaning that the hash table has $2n$ slots.

When inserting a key $x$, we start from cell $h_1(x)$ and look for a cell to place $x$. At cell $i$, we do the following:
    \begin{itemize}
        \item If $i$ is empty, then we place $x$ in it and terminate.
        \item If $i$ contains a key $x'$ for which the fingerprint $h(x)$ precedes $h(x')$ (that is, either $h_1(x)<h_1(x')$ or $h_1(x)=h_1(x')\land h_2(x)<h_2(x')$), then we place $x$ in cell $i$ and remove $x'$. We then start from cell $i+1$ (wrapping around if necessary) and look for a cell to place $x'$ (this ensures that our hash table is \emph{ordered}).
        \item Otherwise, we move on to cell $i+1$  (wrapping around if necessary).
    \end{itemize}
    The process of deleting a key $x$ is similar: We find the cell $i$ that stores the fingerprint $h(x)$ (if multiple cells store $h(x)$, choose any of them) and delete the content of $i$. After this, we move fingerprints to the left to maintain the structure of the hash table: if the fingerprint stored at cell $i+1$ has hash $\le i$ (i.e., it starts from the left of cell $i$), then we move that fingerprint to cell $i$, and continue this process for cell $i+2$. Since our hash table has load factor $1/2$, all these operations can be done in expected constant time.

    To store this hash table, a na\"{\i}ve implementation would use $O(n\log (n\delta^{-1}))$ bits of space by directly storing a fingerprint in each cell. The standard quotient filter improves upon this by noting that we do not need to explicitly store the hashes of the keys, since most information about the hashes is already implicitly known (i.e., a key is placed in a cell not far away from its hash).

    Formally, in each cell $i$, a quotient filter stores the following. We explicitly store the \emph{remainder} of the key placed in $i$ (hence the name \emph{quotient} filter). We also store three extra bits of metadata for each cell $i$ for recovering the hashes:
    \begin{enumerate}
        \item One bit indicating whether the cell is empty.
        \item One bit indicating whether the key in cell $i$ has the same hash as that of $i-1$.
        \item One bit indicating whether any key in the hash table has $i$ as its hash.
    \end{enumerate}
    This metadata can be maintained in a straightforward way when performing updates. To recover the hash of the key placed at cell $i$, we find the first empty cell $i_0$ to the left of $i$ using the type 1 metadata. We then count the number of distinct hashes among the keys stored in $(i_0,i]$ using the type 2 metadata. Finally, say the number is $k$, we use the type 3 metadata to find the $k$-th cell $i_k$ in $(i_0,i]$ that has some key hashed to it, which is the hash of the key at cell $i$.

    In total, the standard implementation of a quotient filter, assuming $\beta= 1$, uses $O(n(1+\max\{0,\log\delta^{-1}\}))$ bits, while supporting $O(1)$ expected-time operations.

    \paragraph{The improved filter.} Our improvement is extremely simple: We partition the hash table into \defn{groups} of $g$ cells each (except one group which may be smaller), where $g=\lceil \eta w/(1+\max\{0,\log\delta^{-1}\})\rceil$ for a small constant $\eta>0$. We store the hash table using an entropy-encoded array \cref{thm:arithmetic-code-poly-n} of $n'=\lceil 2n/g\rceil$ entries, where each entry corresponds to a group of the hash table, and the value of this entry is just the concatenation of all remainders and metadata in the group. We will choose $\eta$ to be sufficiently small, such that each entry contains $<(2/3)\log n$ bits of information, which means that our alphabet size is only $O(n^{2/3})\ll n$, so we can ignore the effect of the alphabet size on our data structure. In the remainder of this proof, we show that our improved filter is indeed space-efficient.

    \paragraph{Notation.} We say that a key $x$ \defn{overflows} if $x$ is placed in a cell that belongs to a different group from its hash $h_1(x)$. For every cell $i$, define $S_i$ to be the multiset of remainders that correspond to keys that hash to $i$. Define $S^*_i\subseteq S_i$ to be the submultiset of remainders corresponding to the non-overflown keys, i.e., keys that are placed in the same group as cell $i$.

    \paragraph{The analysis.} Our goal is to bound the entropy $H=\sum_{j=1}^{n'}\log(n'/f_{A_j})$ of the array, where $A_j$ is the value of the $j$-th entry (interpreted as a bit string) and $f_{A_j}$ is the number of occurrences of $A_j$ in $A_1,\dots,A_{n'}$.

    We assume that $g$ divides $2n$, that is, each group contains exactly $g$ cells (if not, we directly bound the cost of the last group as $\log n$). In this case, viewing the $A_j$'s as random variables, we have that they are identically distributed (although not independent). Since $H$ is the cost of encoding the $A_j$'s under the empirical distribution, if we fix some specific distribution $D$, then it is always the case that $H\le \sum_{j=1}^{n'}-\log D(A_j)$. Since the $A_j$'s are identically distributed, we can let $D$ be the distribution of $A_1$, which gives us $H\le \sum_{j=1}^{n'}H(A_j)$, where $H(A_j)$ is the binary entropy of $A_j$. It now remains to bound $\sum_{j=1}^{n'}H(A_j)$. In the following, we present two different bounds for $\sum_{j=1}^{n'}H(A_j)$ that, when combined, implies our space bound.

    \paragraph{Large $\log\epsilon^{-1}$ case.} When $\log\epsilon^{-1}\ge \sqrt{\log n}$, we bound $H(A_j)$ as follows: To describe $A_j$, we can directly send the metadata bits, and then send the remainders of the nonempty cells. This shows that $H(A_j)$ is at most $O(g)+\log(2\delta^{-1}+O(1))\cdot (\text{\# of keys in the $j$-th group})$. Summing this over all $O(n/g)$ groups, the total entropy is at most
    \begin{align*}
        H\le n\log\delta^{-1}+O(n)=n\log\epsilon^{-1}+O(n),
    \end{align*}
    which is at most $(1+O(1/\sqrt{\log n}))\cdot n\delta^{-1}H(\text{Pois}(\delta))$ since $n\delta^{-1}H(\text{Pois}(\delta))\ge n\log\epsilon^{-1}\ge n\sqrt{\log n}$ (by \cref{clm:connect_delta_with_eps}). Note that we always have $\epsilon<0.9$ when $n$ is sufficiently large.

    \paragraph{Small $\log\epsilon^{-1}$ case.} When $\log\epsilon^{-1}<\sqrt{\log n}$, we consider the following information denoted as $B_j$:
    \begin{itemize}
        \item The size $|S_i|$ of every cell $i$ in the $j$-th group.
        \item The number of overflown keys placed in the $j$-th group, as well as their fingerprints and the cells they are placed in. Denote this information as $Y_j$.
        \item The multisets $S^*_i$ of every cell $i$ in the $j$-th group.
        \item One extra bit, which is the type 2 metadata bit of the first cell in the $j$-th group.
    \end{itemize}

    Note that $B_j$ determines $A_j$, which means that $H(B_j)\ge H(A_j)$:
    \begin{itemize}
        \item After knowing the overflown keys and the sizes $|S_i|$, we can simulate the insertions of the ordered hash table, and know whether each cell in the $j$-th group contains a key, and if so, we know this key's hash. Using this information, we already know all the metadata bits in the $j$-th group, except for the type 2 metadata bit of the first cell in the $j$-th group, which is given explicitly in $B_j$.
        \item It remains to know the remainders of the keys. For the overflown keys, this is given explicitly in $B_j$. For the non-overflown keys, we determine the remainders using the sets $S^*_i$. For each cell $i$, we know the location of all non-overflown keys with hash $i$. Since our hash table is ordered, it must be the case that the remainder of the $k$-th such key is the $k$-th smallest value in $S^*_i$. Thus, $B_j$ determines all remainders.
    \end{itemize}
    
    It remains to bound $\sum_{j=1}^{n'}H(B_j)$. For each $B_j$, we have that
    \begin{align}
        \label{equ:B_entropy_bound}
        H(B_j)\le\bk*{\sum_{i\text{ in j-th group}}H(|S_i|)}+H(Y_j)+\bk*{\sum_{i\text{ in j-th group}}H(S^*_i\mid \{|S_{i'}|\}_{i'\text{ in j-th group}},Y_j)}+1.
    \end{align}
    Given $\{|S_{i'}|\}_{i'\text{ in j-th group}}$ and $Y_j$, one knows the hashes of all keys placed in the $j$-th group, and can therefore determine the sizes $|S^*_i|$. Thus, for every cell $i$ in the $j$-th group, we have
    \begin{align*}
        H(S^*_i\mid \{|S_{i'}|\}_{i'\text{ in j-th group}},Y_j)\le H(S^*_i\mid |S_i|,|S^*_i|).
    \end{align*}
    Plugging this into \eqref{equ:B_entropy_bound} gives
    \begin{align}
        \label{equ:B_entropy_bound_2}
        H(B_j)&\le H(Y_j)+\sum_{i\text{ in $j$-th group}}\big(H(|S_i|)+H(S^*_i\mid |S_i|,|S^*_i|)\big)+1\nonumber \\
        &\le H(Y_j)+\sum_{i\text{ in $j$-th group}}\big(H(|S_i|)+H(S_i\mid |S_i|)+H(S^*_i\mid |S_i|,|S^*_i|,S_i)\big)+1\nonumber \\
        &\le H(Y_j)+\sum_{i\text{ in $j$-th group}}H(S_i)+1,
    \end{align}
    where the last line is due to the fact that $S^*_i$ contains the smallest remainders in $S_i$ (due to our hash table being ordered), so $S_i$ and $|S^*_i|$ uniquely determine $S^*_i$.

    Summing \eqref{equ:B_entropy_bound_2} for all $j$, we have that
    \begin{align}
        \label{eq:H_bound}
        H\le\sum_{j=1}^{n'}H(B_j)\le \bk*{\sum_{j=1}^{n'}H(Y_j)}+\bk*{\sum_{i=1}^{2n}H(S_i)}+O(n\log\epsilon^{-1}/\log n).
    \end{align}
    Note that the second term of \eqref{eq:H_bound} is at most $(n\delta^{-1}+O(1))H(\text{Pois}(\delta))$, since the latter is the cost of storing the number of keys of each fingerprint as if they are independent random variables. It remains to bound $\bk*{\sum_{j=1}^{n'}H(Y_j)}$.

    The variable $Y_j$ consists of the following parts:
    \begin{itemize}
        \item The number of overflown keys in the $j$-th group. This costs $\log g=O(\log\log n)$ bits per group, which sums to $O(\log\log n\cdot n/g)=O(n\log\epsilon^{-1}\log\log n/\log n)$ bits.
        \item For each overflown key $x$, $Y_j$ contains:
        \begin{itemize}
            \item The cell that $x$ is placed in, costing at most $\log g=O(\log\log n)$ bits.
            \item The hash of $x$, costing $O(\log\log n)$ bits in expectation since $x$ is placed at distance $O(\log n)$ from its hash with high probability (this is a property of linear probing hash tables).
            \item The remainder of $x$, costing $\log(2\delta^{-1}+O(1))=O(\log\epsilon^{-1})$ bits.
        \end{itemize}
    \end{itemize}
    Overall, $\bk*{\sum_{j=1}^{n'}H(Y_j)}$ costs
    \begin{align}
        \label{equ:overflown_entropy}
        O(n\log\epsilon^{-1}\log\log n/\log n)+O(\log\log n+\log\epsilon^{-1})\cdot \E[\text{\# of overflown keys}]
    \end{align}
    in expectation. For any key $x$ in a linear probing hash table, the expected distance between the cell that $x$ is placed in and the hash of $x$ is $O(1)$. Thus, since the hash function $h$ is assumed to be fully random and each group contains $g$ cells, the probability that $x$ overflows is $O(1/g)$. In total, we have $O(n/g)$ overflown keys in expectation. Plugging this back into \eqref{equ:overflown_entropy} gives $O(n\log\epsilon^{-1}\cdot (\log\log n+\log \epsilon^{-1})/\log n)$.
    
    In summary, we can bound the expected entropy $H$ of the array $A_1,\dots,A_{n'}$ as
    \begin{align*}
        H&\le (n\delta^{-1}+O(1))H(\text{Pois}(\delta))+O(n\log\epsilon^{-1}\cdot (\log\log n+\log \epsilon^{-1})/\log n).
    \end{align*}
    Since $\log\epsilon^{-1}\le \sqrt{\log n}$, the latter term is at most
    \begin{align*}
        O\bk*{\frac 1{\sqrt{\log n}}}\cdot n\log\epsilon^{-1}=O\bk*{\frac 1{\sqrt{\log n}}}\cdot n\delta^{-1}H(\text{Pois}(\delta)).
    \end{align*}
    For $\epsilon<0.9$, this is by \cref{clm:connect_delta_with_eps}. For $0.9\le \epsilon<1$, both $n\log\epsilon^{-1}$ and $n\delta^{-1}H(\text{Pois}(\delta))$ are $\Theta(n)$.
    \paragraph{Combining two cases.} We have shown that
    \begin{align*}
        H&\le\bk*{1+O\bk*{\frac 1{\sqrt{\log n}}}}\cdot n\delta^{-1}H(\text{Pois}(\delta))
    \end{align*}
    in expectation whenever $1-\epsilon=\Omega(1)$. Therefore, when we apply \cref{thm:arithmetic-code-poly-n} to store the array $A_1,\dots,A_{n'}$, our data structure only uses

    \begin{align*}
        &\bk*{1+O\bk*{\frac {\log\log n}{\log n}}}H+O(n/\poly\log n) \\
        \le{}&\bk*{1+O\bk*{\frac 1{\sqrt{\log n}}}}\cdot n\delta^{-1}H(\text{Pois}(\delta))
    \end{align*}
    bits of space in expectation. This concludes our analysis.
\end{proof}

%% file: app_poisson.tex
\section{Proof of Claim \ref{clm:connect_delta_with_eps}}
\label{app:poisson}

\Poisson*

Recall that $\delta$ is defined as $\ln(1/(1-\epsilon))$.

\begin{proof}    
    By definition of Poisson distributions, we have
    \begin{align}
        \label{eq:bound_1}
        n\delta^{-1}H(\text{Pois}(\delta))=n\log\epsilon^{-1}+n\log e-n\cdot\bk*{\log(\epsilon^{-1}\delta)-\delta^{-1}e^{-\delta}\sum_{k=0}^{\infty}\frac{\delta^k\log(k!)}{k!}}.
    \end{align}
    For the upper bound, we study the infinite summation
    \begin{align}
        \label{eq:bound_2}
        \delta^{-1}e^{-\delta}\sum_{k=0}^{\infty}\frac{\delta^k\log(k!)}{k!}.
    \end{align}
    Note that $\log(k!)=0$ when $k<2$, so we only have to consider the terms where $k\ge 2$. In this case, we have that
    \begin{claim}
        For any $k\ge 2$, $\log(k!)\le k(k-1)/2$.
    \end{claim}

    \begin{proof}
        Exponentiating both sides, we get
        \begin{align*}
            k!\le 2^{k(k-1)/2}=2^{0+1+\cdots+k-1},
        \end{align*}
        which holds because $k\le 2^{k-1}$ for any $k\ge 1$.
    \end{proof}

    Plugging this into \eqref{eq:bound_2}, we obtain
    \begin{align*}
        \delta^{-1}e^{-\delta}\sum_{k=0}^{\infty}\frac{\delta^k\log(k!)}{k!}&\le (1/2)\delta^{-1}e^{-\delta}\sum_{k=2}^{\infty}\frac{\delta^k}{(k-2)!} \\
        &\le (1/2)\delta e^{-\delta}\sum_{k=0}^{\infty}\frac{\delta^k}{k!}=(1/2)\delta.
    \end{align*}
    Plugging this back in \eqref{eq:bound_1}, it remains to show that, for $0<\delta<\ln 10$,
    \begin{align*}
        &n\log\epsilon^{-1}+n\log e-n\cdot\bk*{\log(\epsilon^{-1}\delta)-(1/2)\delta}\le n\log\epsilon^{-1}+n\log e \\
        \Leftrightarrow{}&\log(\epsilon^{-1}\delta)-(1/2)\delta=-\log(1-e^{-\delta})+\log(\delta)-(1/2)\delta\ge 0,
    \end{align*}
    which can be verified using simple techniques and is omitted.

    For the lower bound, we can directly ignore the infinite summation in \eqref{eq:bound_1}, and only have to show that
    \begin{align*}
        &n\log\epsilon^{-1}+n\log e-n\log(\epsilon^{-1}\delta)\ge n\log\epsilon^{-1}\\
        \Leftrightarrow{}&\log e-\log(\epsilon^{-1}\delta)\ge 0.
    \end{align*}
    It is easy to verify that this holds when $0<\epsilon<0.9$, so we omit the verification.
\end{proof}